\documentclass[11pt]{article}

\usepackage[margin=1in]{geometry}

\usepackage{amstext,amssymb,amsmath}
\usepackage{relsize,xspace,ifthen}
\usepackage[T1]{fontenc}
\usepackage{times}
\usepackage{microtype}
\usepackage{enumitem} 
\usepackage[utf8]{inputenc}
\usepackage{tikz}
\usetikzlibrary{arrows}
\usetikzlibrary{shapes.geometric}
\usetikzlibrary{calc}
\usepackage{hyperref}
\usepackage{comment}

\usepackage[thmmarks, amsmath, hyperref]{ntheorem}

\newenvironment{cenv}{\begin{list}{}{%
      \setlength{\labelwidth}{1.5em}%
      \setlength{\leftmargin}{\labelwidth}%
      \addtolength{\leftmargin}{\labelsep}%
      \setlength{\listparindent}{0em}%
      \setlength{\topsep}{10pt}%
      \setlength{\itemsep}{5pt}%
      \setlength{\parsep}{0pt}%
    }
  }{
  \end{list}
}

\newcounter{claimcounter}

\newenvironment{Claim}{
  
  \refstepcounter{claimcounter}
  \begin{cenv}
  \item[{Claim \arabic{claimcounter}.}]
  }{
  \end{cenv}
}

\newenvironment{ClaimProof}[1][]{\noindent{%
\ifthenelse{\equal{#1}{}}{{\sl Proof.\ }}{{\sl #1.\ }}%
}}{\hspace*{1em}\nobreak\hfill$\dashv$\endtrivlist\addvspace{2ex plus
0.5ex minus0.1ex}}

\usepackage{algorithm}
\usepackage{algpseudocode}
\usepackage{eqparbox}

\newtheorem{theorem}{Theorem}[section]
\newtheorem{lemma}[theorem]{Lemma}
\newtheorem{corollary}[theorem]{Corollary}
\newtheorem{definition}[theorem]{Definition}
\newtheorem{remark}[theorem]{Remark}

\numberwithin{equation}{section}

\newcommand{\minor}{\preceq}
\newcommand{\st}{\mathrel : }
\newcommand{\N}{{\mathbb N}}

\newcommand{\CCC}{\mathcal{C}}

 \newcommand{\VVV}{\mathcal{V}}
 \newcommand{\XXX}{\mathcal{X}}
\newcommand{\YYY}{\mathcal{Y}} 
\newcommand{\norm}[1]{\ensuremath{\left\lVert#1\right\rVert}}
 
\newcommand{\AWs}{\textup{AW}$[*]$\xspace}

\newcommand{\FO}{\textup{FO}}

\newcommand{\Oof}{\mbox{$\cal O$}}
\newcommand{\rad}{\mathrm{rad}}
\newcommand{\aug}{\mathrm{aug}}
\newcommand{\wcol}{\mathrm{wcol}}
\newcommand{\wreach}{\mathrm{WReach}}

\newcommand{\tp}{\operatorname{tp}^+}
\newcommand{\atp}{\operatorname{atp}^+}
\newcommand{\itp}{\operatorname{itp}^+}
\newcommand{\dist}{\operatorname{dist}}

\newcommand{\NOdM}{Ne{\v s}et{\v r}il and Ossona de Mendez}

\newcommand{\npprob}[5]{%
\begin{center}\normalfont\fbox{%
\begin{tabular}[t]{rp{#1}}%
\multicolumn{2}{l}{\textsc{#2}}\\%
\textit{Input:} & #3\\%
\textit{Problem:} & #5%
\end{tabular}}%
\end{center}}

\usepackage{color}
\definecolor{gruen}{rgb}{0,0.6,0.2}
\newcommand{\red}{\color{red}}

\newcommand{\green}{\color{gruen}}
\newcounter{rbcounter}
\newcommand{\randbem}[3]{\stepcounter{rbcounter}\parbox{0mm}{\hspace*{-0.3em}#1{}$^{\arabic{rbcounter}}$}\marginpar{#1\raggedright\footnotesize\textbf{#2$^{\arabic{rbcounter}}$: }#3}}

\newcommand{\stephan}[1]{\randbem{\red}{St}{#1}}
\newcommand{\sebastian}[1]{\randbem{\green}{Se}{#1}}

\renewcommand{\phi}{\varphi}
\renewcommand{\epsilon}{\varepsilon}

\newcommand{\case}[1]{\par\medskip\noindent\textit{Case #1: }}

\newenvironment{cs}{
  \begin{description}
    \renewcommand{\case}[1]{\item[\itshape\mdseries Case ##1:]}
  }{
  \end{description}
}

\newcommand{\tup}[1]{\bar{#1}}

\def\qed{\hspace*{\fill}$\Box$}
\newenvironment{proof}[1][]
{\ifthenelse{\equal{#1}{}}{\noindent\textit{Proof.
    }}{\noindent\textit{#1. }}}%
{\qed\par\smallskip}

\newcommand{\qr}{\operatorname{qr}}

\newcommand{\dom}{\operatorname{dom}}
\newcommand{\rg}{\operatorname{rg}}

\newcommand{\hphi}{\widehat{\phi}}

\newcommand{\colprod}{\,\tikz{\path[draw,fill=red!20] circle (2.5pt);}\,}

\begin{document}

\title{Deciding first-order
  properties of nowhere dense graphs}


\author{Martin Grohe\\RWTH Aachen
  University\\grohe@informatik.rwth-aachen.de \and Stephan
  Kreutzer\\Technical University Berlin\\stephan.kreutzer@tu-berlin.de
  \and Sebastian Siebertz\\Technical University Berlin\\sebastian.siebertz@tu-berlin.de}

\date{}

\maketitle

\begin{abstract}
  Nowhere dense graph classes, introduced by Ne\v set\v ril and Ossona
  de Mendez~\cite{NesetrilO11}, form a large variety of classes of
  ``sparse graphs'' including the class of planar graphs, actually all
  classes with excluded minors, and also bounded degree graphs and
  graph classes of bounded expansion.

  We show that deciding properties of graphs definable in first-order
  logic is fixed-parameter tractable on nowhere dense graph
  classes. At least for graph classes closed under taking subgraphs,
  this result is optimal: it was known before that for all
  classes~$\mathcal C$ of graphs closed under taking subgraphs, if
  deciding first-order properties of graphs in~$\mathcal C$ is
  fixed-parameter tractable, then~$\mathcal C$ must be nowhere dense
  (under a reasonable complexity theoretic assumption).

  As a by-product, we give an algorithmic construction of sparse
  neighbourhood covers for nowhere dense graphs. This extends and
  improves previous constructions of neighbourhood covers for graph
  classes with excluded minors. At the same time, our construction is
  considerably simpler than those.

Our proofs are based on a new game-theoretic characterisation of
  nowhere dense graphs that allows for a recursive version of
  locality-based algorithms
  on these classes. On the logical side, we prove a
  ``rank-preserving'' version of Gaifman's locality
  theorem.
\end{abstract}

\section{Introduction}

Algorithmic meta theorems attempt to explain and unify algorithmic
results by proving tractability not only for individual problems, but
for whole classes of problems. These classes are typically defined in
terms of logic. The meaning of ``tractability'' varies; for example,
it may be linear or polynomial time solvability, fixed-parameter
tractability, or polynomial time approximability to some ratio. The
prototypical example of an algorithmic meta theorem is Courcelle's
Theorem~\cite{cou90}, stating that all properties of graphs of bounded
tree-width that are definable in monadic second-order logic are
decidable in linear time. Another well-known example is Papadimitriou
and Yannakakis's \cite{papyan91} result that all optimisation problems
in the class MAXSNP, which is defined in terms of a fragment of
existential second-order logic, admit constant-ratio polynomial time
approximation algorithms. By now, there is a rich literature on
algorithmic meta theorems (see, for example,
\cite{bodfomlok+09,coumakrot00,coumakrot01,dawgrokre07,
  dawgrokre+06,dvokratho10,frigro01,kretaz10a,kretaz10,see96} and the
surveys \cite{gro07b,grokre11,kre11}).  While the main motivation for
proving such meta theorems may be to understand the ``essence'' and
the scope of certain algorithmic techniques by abstracting from
problem-specific details, sometimes meta theorems are also crucial for
obtaining new algorithmic results. A recent example is the quadratic
time algorithm for a structural decomposition of graphs with excluded
minors from \cite{grokawree13}, which builds on Courcelle's Theorem in
an essential way. Furthermore, meta theorems often give a quick and
easy way to see that certain problems can be solved efficiently (in
principle), for example in linear time on graphs of bounded
tree-width. Once this has been established, a problem specific
analysis may yield better algorithms -- even though implementations
of, for instance, Courcelle's theorem have shown that the direct
application of meta theorems
can yield competitive algorithms for common problems such as the
dominating set problem (see~\cite{LangerRRS12}).

In this paper, we prove a new meta theorem for first-order logic on
nowhere dense classes of graphs. These classes were introduced by Ne\v
set\v ril and Ossona de Mendez~\cite{NesetrilOdM12,NesetrilO11} as a
formalisation of classes of ``sparse'' graphs. All familiar examples
of sparse graph classes, like the class of planar graphs, classes of
bounded tree-width, classes of bounded degree, and indeed all classes
with excluded topological subgraphs are nowhere
dense. Figure~\ref{fig:classes} shows the containment relations
between these and other sparse graph classes.\footnote{Notably,
  classes of bounded average degree or bounded degeneracy are not
  necessarily nowhere dense. To be precise: for every~$k\ge 2$ the
  class of all graphs of degeneracy at most~$k$ is somewhere
  dense. This is reasonable, because every graph can be turned into a
  graph of degeneracy~$2$ by simply subdividing every edge
  once. Recall that a graph has \emph{degeneracy} at most~$d$ if every
  subgraph has a vertex of degree at most~$d$. Degeneracy at most~$d$
  implies that the graph and all its subgraphs have average degree at
  most~$2d$ and hence have a linear number of edges. Contrarily,
  graphs in nowhere dense classes can have an edge density
  of~$n^{1+\epsilon}$ and are therefore not necessarily degenerate.}
``Nowhere density'' turns out to be a very robust concept with several
seemingly unrelated natural characterisations (see
\cite{NesetrilOdM12,NesetrilO11}). Furthermore, Ne\v set\v ril and
Ossona de Mendez~\cite{NesetrilO11} established a clear-cut dichotomy
between nowhere dense and somewhere dense graph classes. The exact
definition of nowhere dense graph classes is
technical 
and we defer it to Section~\ref{sec:nowheredense}.

\begin{figure}[t]
  \centering
  \begin{tikzpicture}
     [
     klasse/.style = {draw, fill=gruen!20,shape=ellipse}, >=stealth
     ]
     \small

      \shade[top color=red!20] (-3,10.5) -- (-3,9.5) .. controls (3,7.8) and (6,7.8)
     .. (11,9.5) -- (11,10.5) -- cycle;

      \shade[top color=white,bottom color=gruen!20] (-3,-0.5) -- (-3,7) .. controls (3,8.7) and (6,8.7)
     .. (11,7) -- (11,-0.5) -- cycle;

    \draw (4.5,0) node[klasse] (pl) {planar}
               (4.5,1.2) node[klasse] (bg) {bounded genus}
               (0,0.75) node[klasse] (tw) {bounded tree-width}
               (-1.2,2.7) node[klasse,text width=2cm,text centered] (ltw) {bounded local tree-width}
               (3,2.9) node[klasse] (x) {excluded minor}
               (6,4) node[klasse] (xt) {excluded topological subgraph}
               (9,1.2) node[klasse] (d) {bounded degree}
               (6,5.25) node[klasse] (bx) {bounded expansion}
               (0,4.5) node[klasse] (lx) {locally excluded minor}
               (2,6.25) node[klasse] (lbx) {locally bounded expansion}
               (3.5,7.5) node[klasse] (nd) {nowhere dense}
               (6,9.3) node[klasse,fill=red!20] (bd) {bounded degeneracy}
     ;
     
     \draw[->] (pl) edge (bg) (bg) edge (x) (tw) edge (ltw) edge (x)
     (ltw) edge (lx) (x) edge (lx)
     edge (xt) (lx) edge (lbx) (lbx) edge (nd) (d) edge (xt) (xt) edge
     (bx) (bx) edge (lbx) (bx) edge (bd);

      \draw[->] (d) edge (ltw) (bg) edge (ltw);

     \draw[gruen] (-3,7) .. controls (2.5,8.5) and (5.5,8.5)
     .. (11,7);
     \path (-2.5,-0.2) node[anchor=west] {\color{gruen}\bfseries nowhere dense};
     
     \draw[red] (-3,9.5) .. controls (2.5,7.9) and (5.5,7.9) .. (11,9.5);
     \path (-2.5,10.2) node[anchor=west] {\color{red}\bfseries somewhere dense};
  \end{tikzpicture}

  \caption{Sparse graph classes}
  \label{fig:classes}
\end{figure}
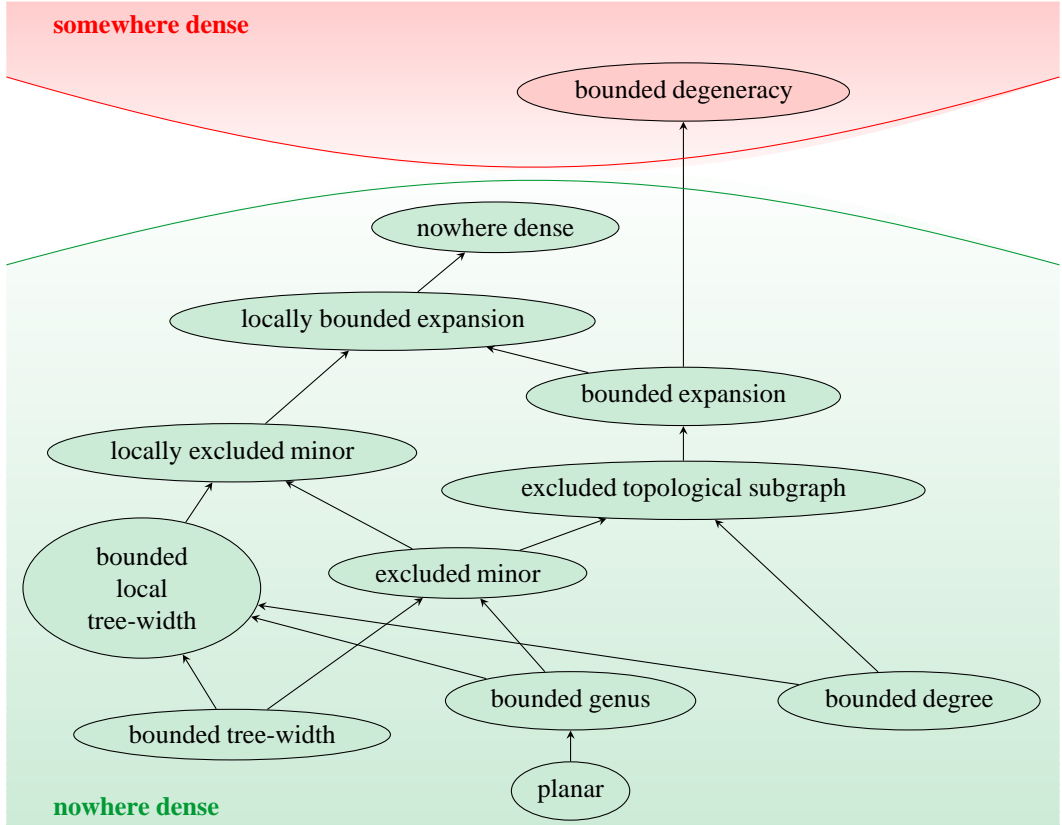

\begin{theorem}\label{theo:main}
  For every nowhere dense class~$\mathcal C$ and every~$\epsilon>0$,
  every property of graphs definable in first-order logic can be
  decided in time~$\Oof(n^{1+\epsilon})$ on~$\mathcal C$.
\end{theorem}

In particular, deciding first-order properties is fixed-parameter
tractable on nowhere dense graph classes.\footnote{There is a minor
  issue regarding non-uniform vs uniform fixed-parameter tractability,
  see Remark~\ref{rem:uniform}.}  Deciding first-order properties of
arbitrary graphs is known to be complete for the parameterized
complexity class~$\textup{AW}[*]$ and thus unlikely to be
fixed-parameter tractable~\cite{dowfeltay96}.

Ne{\v s}et{\v r}il and Ossona de Mendez \cite{NesetrilOdM12} already
proved that deciding properties definable in existential first-order
logic is fixed-parameter tractable on nowhere dense graphs. Dawar and
Kreutzer~\cite{DawarKre09} showed that dominating set (parameterized
by the size of the solution) is fixed-parameter tractable on nowhere
dense graphs. Our theorem implies new fixed-parameter tractability
results on nowhere dense graphs for many other standard parameterized
problems, for example, connected dominating set and digraph kernel
(both parameterized by the size of the solution), Steiner tree
(parameterized by the size of the tree) and circuit satisfiability
(parameterized by the depth of the circuit and the Hamming weight of
the solution). The last result requires the generalisation of our
theorem from graphs to arbitrary relational structures, which is
straightforward.

Our theorem can be seen as the culmination of a long line of meta
theorems for first order logic. The starting point is Seese's
\cite{see96} result that first-order properties of bounded degree
graphs can be decided in linear time. Frick and Grohe~\cite{frigro01}
gave linear time algorithms for planar graphs and all apex-minor-free
graph classes and~$\Oof(n^{1+\epsilon})$ algorithms for graphs of
bounded local tree-width. Flum and Grohe~\cite{flugro01} proved that
deciding first-order properties is fixed-parameter tractable on graph
classes with excluded minors, and Dawar, Grohe, and
Kreutzer~\cite{dawgrokre07} extended this to classes of graphs locally
excluding a minor. Finally, Dvo\v r\'ak, Kr\'al, and
Thomas~\cite{dvokratho10} proved that first-order properties can be
decided in linear time on graph classes of bounded expansion and in
time~$\Oof(n^{1+\epsilon})$ on classes of locally bounded
expansion. All these classes are nowhere dense, and there are nowhere
dense classes that do not belong to any of these classes. For example,
the class of all graphs whose girth is larger than the maximum degree
is nowhere dense, but has unbounded expansion. If to every graph in
this class we add one vertex and connect it with all other vertices,
we obtain a class of graphs that is still nowhere dense, but does not
even have locally bounded expansion. However, what makes our theorem
interesting is not primarily that it is yet another extension of the
previous results, but that it is optimal for classes $\mathcal C$
closed under taking subgraphs: under the standard complexity theoretic
assumption~$\textup{FPT}\neq\textup{W[1]}$, Kreutzer \cite{kre11} and
Dvo\v r\'ak et al.~\cite{dvokratho10} proved that if a class~$\mathcal
C$ closed under taking subgraphs is somewhere dense (that is, not
nowhere dense), then deciding first-order properties of graphs
in~$\mathcal C$ is not fixed-parameter tractable. Note that all
classes considered in the previous results are closed under taking
subgraphs.  Hence our result supports the intuition that nowhere dense
classes are the natural limit for many algorithmic techniques for
sparse graph classes.

Technically, we neither use the structural graph theory underlying
\cite{dawgrokre07,flugro01} nor the quantifier elimination techniques
employed by \cite{dvokratho10}. Our starting point is the locality
based technique introduced in \cite{frigro01}. In a nutshell, this
technique works as follows. Using Gaifman's theorem, the problem to
decide whether a general first-order formula~$\phi$ is true in a graph
can be reduced to testing whether a formula is true
in~$r$-neighbourhoods in the graph, where the radius~$r$ only depends
on~$\phi$, and solving a variant of the (distance~$d$) independent set
problem. Hence, if~$\CCC$ is a class of graphs
where~$r$-neighbourhoods have a simple structure, such as the class of
planar
graphs or classes of bounded local tree-width, this method gives
an easy way for deciding properties definable in first-order
logic.

Applying this technique to nowhere dense classes of graphs immediately
runs into problems, as~$r$-neighbourhoods in nowhere dense graphs do
not necessarily have a simple structure that can be exploited algorithmically. We therefore
iterate the locality based approach. Using locality we reduce the
first-order model-checking problem to the problems of evaluating
formulas in~$r$-neighbourhoods and solving a variant of the
independent set problem. We then show that~$r$-neighbourhoods~$N$ in
nowhere dense graphs can be split by deleting a set~$W$ of only a few
vertices into smaller neighbourhoods. We apply the locality
argument again and transform our formula into formulas to be evaluated
in $r$-neighbourhoods in~$N-W$ and solving the independent set
problem on~$N-W$. We show that on nowhere dense classes of graphs
this process terminates after a constant number of steps.

The three main steps of our proof, each of which may be of independent
interest, are the following.
\begin{itemize}
\item An algorithmic construction of \emph{sparse neighbourhood
    covers} for nowhere dense graphs (Section~\ref{sec:covers}). The
  parameters are surprisingly good: we can cover
  all~$r$-neighbourhoods with sets (called \emph{clusters}) of
  radius~$2r$ such that each vertex is contained in~$n^{o(1)}$
  clusters. For classes of bounded expansion (see
  Figure~\ref{fig:classes}), we even get such covers where each vertex
  is only contained in a constant number of clusters.  In particular,
  the small radius of the clusters substantially improves known
  results for planar graphs and graphs with excluded
  minors~\cite{abrgavmal+07,buscoslaf+07}, which all have bounded
  expansion.
\item A new characterisation of nowhere dense graph classes in terms
  of a game, the \emph{Splitter Game} (Section~\ref{sec:game}). We use
  this game to formalise the process of localising and splitting
  described above and showing that it terminates on nowhere dense
  graphs. It turns out that it only terminates on nowhere dense
  graphs, thus providing a necessary and sufficient condition for
  nowhere density.
\item A \emph{Rank-Preserving Locality Theorem}
  (Section~\ref{sec:local}), strengthening Gaifman's well-known
  locality theorem for first-order logic by translating first-order
  formulas into local formulas of the same rank. The key innovation
  here is a new, discounted rank measure for first-order formulas.
\end{itemize}

We describe the main algorithm proving Theorem~\ref{theo:main} in
Section~\ref{sec:mc}.
\parskip0pt

\section{Preliminaries}
\label{sec:prelim}

We assume familiarity with basic concepts of graph theory and refer to
\cite{Diestel05} for background. We denote the set of positive
integers by~$\N$. For~$k\in\N$ we write~$[k]$ for the set~$\{1,\ldots,
k\}$. We will often write~$\tup{a}$ for a~$k$-tuple~$(a_1,\ldots,
a_k)$ and~$a\in\tup{a}$ for~$a\in\{a_1,\ldots, a_k\}$.

In this section, we will review the necessary background from graph
theory and parameterized complexity theory. We will provide some
background on logic in Section~\ref{sec:local}.

\smallskip

\noindent\textbf{Background from graph theory. } All graphs in this
paper are finite and simple, i.e., they do not have loops or multiple
edges between the same pair of vertices. Whenever we speak of a graph
we mean an undirected graph and we will explicitly mention when we
deal with directed graphs.

If~$G$ is a graph then~$V(G)$ denotes its set of vertices and~$E(G)$
its set of edges. We write~$n := |V(G)|$ for the \emph{order} of $G$.

An \emph{orientation} of~$G$ is a directed graph~$\vec{G}$ on the same
vertex set, which is denoted~$V(\vec{G})$, such that for each
edge~$\{u,v\}\in E(G)$ the set of arcs~$E(\vec{G})$ contains exactly
one of the arcs~$(u,v)$ or~$(v,u)$. For~$v\in V(\vec{G})$, the
set~$N^-(v):=\{u : (u,v)\in E(\vec{G})\}$ denotes the
\emph{in-neighbours} of~$v$ and~$N^+(v):=\{w : (v,w)\in E(\vec{G})\}$
denotes the \emph{out-neighbours} of~$v$.  The
\emph{indegree}~$d^-(v)$ of a vertex~$v$ is the number in-neighbours
of~$v$. We denote the \emph{maximum indegree} of~$\vec{G}$
by~$\Delta^-(\vec{G})$. For any directed graph~$\vec{G}$ we denote the
underlying undirected graph by~$G$.

We assume that all graphs are represented by adjacency lists so that
the total size of the representation of a graph is linear in the
number of edges and vertices. In fact we will often store an
orientation~$\vec{G}$ of a graph~$G$ and use one adjacency list for
the in-neighbours and one adjacency list for the out-neighbours of
each vertex. This representation allows to check adjacency of vertices
in time~$\Oof(\Delta^-(\vec{G}))$.

For a set~$X\subseteq V(G)$ we write~$G[X]$ for the subgraph of~$G$
induced by~$X$ and we let~$G\setminus X:=G[V(G)\setminus
X]$. For~$k\in\N$,~$G$ is \emph{$k$-degenerate} if for
each~$X\subseteq V(G)$ the graph~$G[X]$ contains a vertex of degree at
most~$k$. If a graph~$G$ is~$k$-degenerate then~$G$ contains at
most~$k\cdot n$ edges and an orientation~$\vec{G}$ of~$G$
with~$\Delta^-(G)\leq k$ can be computed in time~$\Oof(k\cdot n)$ by a
simple greedy algorithm.

The \emph{distance}~$\dist^G(u,v)$ between two vertices~$u,v\in V(G)$
is the length of a shortest path from~$u$ to~$v$ if such a path exists
and~$\infty$ otherwise. The \emph{radius}~$\rad(G)$ of~$G$
is $\min_{u\in V(G)}\max_{v\in V(G)}\dist^G(u,v)$. A vertex~$u\in
V(G)$ such that~$\max_{v\in V(G)}\dist^G(u,v)=\rad(G)$ is called a
\emph{centre vertex} of~$G$.

By~$N_r^G(v)$ we denote the \emph{$r$-neighbourhood} of~$v$ in~$G$,
i.e., the set of vertices of distance at most~$r$ from~$v$ in~$G$.  A
set~$W\subseteq V(G)$ is \emph{$r$-independent in~$G$} if
$\dist^G(u,v) > r$ for all distinct~$u,v\in W$. A~$1$-independent
set
is simply called independent.  A set~$W\subseteq V(G)$ is
\emph{$r$-scattered in~$G$} if~$N_r^G(u)\cap N_r^G(w)=\emptyset$ for
all distinct~$u,w\in W$, i.e., if it is~$2r$-independent.

A graph~$H$ is a \emph{minor} of a graph~$G$, written~$H \minor G$, if
$H$ can be obtained from a subgraph of~$G$ by contracting
edges. Equivalently,~$H$ is a minor of~$G$ if there is a map that
associates with every vertex~$v \in V(H)$ a tree~$T_v \subseteq G$
such that~$T_u$ and~$T_v$ are disjoint for~$u \neq v$ and whenever
there is an edge~$\{u, v\} \in E(H)$ there is an edge in~$G$ between
some node in~$T_u$ and some node in~$T_v$.  The subgraphs~$T_v$ are
called \emph{branch sets}. 

Let~$r\in\N$.~$H$ is a \emph{depth-$r$ minor} of~$G$, denoted~$H
\minor_r G$, if $H$ is a minor of~$G$ and this is witnessed by a
collection of branch sets~$\{ T_v \st v \in V(H) \}$, each of which is
a tree of radius at most~$r$.

For~$s\geq 1$ we denote the complete graph on~$s$ vertices by~$K_s$.

\smallskip

\noindent\textbf{Parameterized complexity. }
The complexity theoretical framework we use in this paper is
parameterized complexity theory, see \cite{DowneyF98, FlumGro06}. A
\emph{parameterized problem} is a pair~$(P,\chi)$, where~$P$ is a
decision problem and~$\chi$ is a polynomial time computable function
that associates with every instance~$w$ of~$P$ a positive integer,
called the \emph{parameter}. The \emph{model-checking problem} for
first-order logic on a class~$\CCC$ of graphs is the following
decision problem. Given an~$\FO$-sentence and a graph~$G\in\CCC$,
decide whether~$G$ satisfies~$\phi$, written~$G\models\phi$. The
parameter is~$|\phi|$. We say that the model-checking problem on a
class~$\CCC$ is \emph{fixed-parameter tractable}, or in the complexity
class FPT, if there is an algorithm that decides on input $(G,\phi)$
whether~$G\models\phi$, in time $f(|\phi|)\cdot |V(G)|^{\Oof(1)}$ for
some computable function $f:\N\rightarrow\N$.  The model-checking
problem for first-order logic on the class of all graphs is known to
be complete for the parameterized complexity class~$\AWs$, which is
widely believed to strictly contain the class FPT. Thus, it is widely
believed that model-checking for first-order logic is not
fixed-parameter tractable.

\section{Nowhere Dense Classes of Graphs}
\label{sec:nowheredense}

Nowhere dense classes of graphs were introduced by Ne\v set\v ril and
Ossona de Mendez~\cite{NesetrilOdM12,NesetrilO11} as a formalisation
of classes of ``sparse'' graphs.

\begin{definition}[Nowhere dense classes]
  A class~$\CCC$ of graphs is \emph{nowhere dense} if for every~$r$
  there is a graph~$H_r$ such that~$H_r \not \minor_r G$ for all~$G
  \in \CCC$.
\end{definition}

It is immediate from the definition that if~$\CCC$ excludes a minor
then it is nowhere dense. But note that excluding some graph as a
depth-$r$ minor is a ``local'' condition that is much weaker than
excluding it ``globally'' as a minor.

\begin{remark}\label{rem:uniform}\upshape
  We call a class~$\CCC$ \emph{effectively nowhere dense} if there is
  a computable function~$f$ such that~$K_{f(r)}\not\preceq_r G$ for
  all~$G\in\CCC$. All natural nowhere dense classes are effectively
  nowhere dense, but it is possible to construct artificial classes
  that are nowhere dense, but not effectively so.

  The way Theorem~\ref{theo:main} is stated in the introduction only
  asserts that deciding first-order properties of nowhere dense graphs
  is \emph{non-uniformly} fixed-parameter tractable. That is, for
  every~$\epsilon>0$ and every sentence~$\phi$ of first-order logic
  there is an algorithm deciding the property defined by~$\phi$ in
  time~$\Oof(n^{1+\epsilon})$. This allows for the algorithms for
  different sentences to be unrelated. For effectively nowhere dense
  classes~$\CCC$, we obtain uniform fixed-parameter tractability, that
  is, a single algorithm that, given an $n$-vertex graph~$G\in\mathcal
  C$,~$\epsilon>0$ and a sentence~$\phi$ of first-order logic, decides
  whether~$\phi$ holds in~$G$ in time ~$f(|\phi|,\epsilon)\cdot
  n^{1+\epsilon}$, for some computable function~$f$.~$\hfill\dashv$
\end{remark}

``Nowhere density'' turns out to be a very robust concept with several
seemingly unrelated natural characterisations (see
\cite{NesetrilOdM12,NesetrilO11}).  We will use several different
characterisations, each supporting different algorithmic techniques. In
the rest of this section we will recall the required equivalences.

The following characterization relates nowhere density to sparsity, albeit
sparsity in the liberal sense that the number of edges of an
$n$-vertex graph is~$n^{1+o(1)}$.

\begin{lemma}[Ne{\v s}et{\v r}il-Ossona de Mendez
  \cite{NesetrilO11}]\label{lem:dens1}
  A class~$\CCC$ of graphs is nowhere dense if, and only if, for every
 ~$r\in \N$
  \begin{equation}\label{eq:ndlimit}
    \lim_{n\to\infty}\sup\left\{\left.\frac{\log|E(H)|}{\log|V(H)|}\;\right|\;
      H\preceq_r G \text{
        with }|V(H)|\ge n, G\in\CCC\right\}\le 1.
  \end{equation}
  Here we take~$\frac{\log|E(H)|}{\log|V(H)|}$ to be~$-\infty$ if
  ~$E(H)=\emptyset$, and we take the supremum to be~$0$ if the set is
  empty, that is, if~$\CCC$ contains no graphs of order at least~$n$.
\end{lemma}

Note that the supremum in \eqref{eq:ndlimit} always exists, because
$\frac{\log|E(H)|}{\log|V(H)|}\le 2$ for all~$H$.  The lemma states
that, as~$n$ gets large, the number of edges in all $r$-shallow minors
of~$n$-vertex graphs in~$\CCC$, is $n^{1+o(1)}$. Thus the graphs
in~$\CCC$ are very uniformly sparse: not only the graphs and all their
subgraphs are sparse, but even all graphs that can be obtained from
subgraphs by ``local'' contractions are. As a further justification of
why nowhere dense classes are inherently interesting as a ``limit of
sparse graph classes'', Ne{\v s}et{\v r}il-Ossona de Mendez proved a
trichotomy stating that for all graph classes~$\CCC$, the limit in
\eqref{eq:ndlimit} approaches $0$ or~$1$ or~$2$ as~$r$ goes to
infinity. This means that if a class~$\CCC$ is not nowhere dense, then
in the limit it is really dense.

For our algorithmic purpose, we state the result in a different form
which follows immediately from the proof of
Lemma~\ref{lem:dens1}.

\begin{lemma}\label{thm:densitycharacterization}
  A class~$\CCC$ of graphs is nowhere dense if, and only if, there is
  a function~$f$ such that for every~$r\in \N$ and every~$\epsilon>0$,
  every depth-$r$ minor~$H$ of a graph~$G\in\CCC$ with~$n\geq
  f(r,\epsilon)$ vertices satisfies~$|E(H)|\leq
  n^{1+\epsilon}$. Furthermore,~$\CCC$ is effectively nowhere dense
  if, and only if, the function~$f$ is computable.
\end{lemma}




We close the section with stating another characterisation of nowhere
dense classes that will be used below.

\begin{definition}[Uniformly quasi-wide classes]
  A class~$\CCC$ of graphs is \emph{uniformly quasi-wide} with margin
 ~$s:\N\rightarrow \N$ and~$N:\N\times\N\rightarrow \N$ if for all
 ~$r,k\in\N$, if~$G\in \CCC$ and~$W\subseteq V(G)$ with~$|W|>N(r,k)$,
  then there is a set~$S\subseteq V(G)$ with~$|S|<s(r)$, such that~$W$
  contains an~$r$-scattered set of size at least~$k$ in~$G \setminus S$.
\end{definition}

We call~$\CCC$ effectively uniformly quasi-wide if the margins~$s$ and
$N$ are computable functions.

\begin{lemma}[Ne{\v s}et{\v r}il-Ossona de Mendez
  \cite{NesetrilO11}]\label{thm:ndtouqw}
  A class~$\CCC$ of graphs is (effectively) nowhere dense if, and only
  if, it is (effectively) uniformly quasi-wide.
\end{lemma}


\section{Game theoretic characterisation of nowhere dense classes}
\label{sec:game}
We now provide a new characterisation of nowhere dense classes in
terms of a game.

\begin{definition}[Splitter game]
  Let~$G$ be a graph and
  let~${\ell},m,r>0$. The~$({\ell},m,r)$-\emph{splitter game} on~$G$
  is played by two players, ``Connector'' and ``Splitter'', as
  follows. We let~$G_0:=G$. In round~$i+1$ of the game, Connector
  chooses a vertex~$v_{i+1}\in V(G_i)$. Then Splitter picks a
  subset~$W_{i+1}\subseteq N_r^{G_i}(v_{i+1})$ of size at most~$m$. We
  let~$G_{i+1}:=G_i[N_r^{G_i}(v_{i+1})\setminus W_{i+1}]$. Splitter
  wins if~$G_{i+1}=\emptyset$. Otherwise the game continues
  at~$G_{i+1}$. If Splitter has not won after~${\ell}$ rounds, then
  Connector wins.

  A \emph{strategy} for Splitter is a function~$f$ that associates
  to every partial play $(v_1, W_1, \dots,$ $v_s, W_s)$ with
  associated sequence~$G_0, \dots, G_s$ of graphs and move~$v_{s+1}\in
  V(G_s)$ by Connector a set~$W_{s+1}\subseteq N_r^{G_s}(v_{s+1})$ of
  size at most~$m$.  A strategy~$f$ is a \emph{winning strategy} for
  Splitter in the~$({\ell},m,r)$-splitter game on~$G$ if Splitter wins
  every play in which he follows the strategy~$f$. If Splitter has a
  winning strategy, we say that he \emph{wins}
  the~$({\ell},m,r)$-splitter game on~$G$.
\end{definition}

\begin{theorem}\label{thm:splittergame}
  Let~$\CCC$ be a nowhere dense class of graphs. Then for
  every~$r>0$ there are~${\ell},m>0$, such that for
  every~$G\in\CCC$, Splitter wins the~$({\ell},m,r)$-splitter game
  on~$G$.

  If~$\CCC$ is effectively nowhere dense, then~${\ell}$ and~$m$ can be
  computed from~$r$.
\end{theorem}

\begin{proof}
  As~$\CCC$ is nowhere dense, it is also uniformly
  quasi-wide. Let~$s_\CCC$ and~$N_\CCC$ be the margin of~$\CCC$.
  Let~$r> 0$ and let~${\ell}:=N_\CCC(r,2s_\CCC(r))$ and~$m:=
  {\ell}\cdot (r+1)$. Note that both~${\ell}$ and~$m$ only depend
  on~$\CCC$ and~$r$. We claim that for any~$G\in\CCC$, Splitter wins
  the~$({\ell},m,r)$-splitter game on~$G$.

  Let~$G\in\CCC$ be a graph. In the~$({\ell},m,r)$-splitter game
  on~$G$, Splitter uses the following strategy. In the first round, if
  Connector chooses~$v_1\in V(G_0)$, where~$G_0:=G$, then Splitter
  chooses~$W_1:=\{v_1\}$. Now let~$i>1$ and suppose that~$v_1,\ldots,
  v_i, G_1,\ldots, G_i, W_1,\ldots, W_i$ have already been
  defined. Suppose Connector chooses~$v_{i+1}\in V(G_i)$. We
  define~$W_{i+1}$ as follows. For each~$1\leq j\leq i$, choose a
  path~$P_{j,i+1}$ in~$G_{j-1}[N_r^{G_{j-1}}(v_j)]$ of length at
  most~$r$ connecting~$v_j$ and~$v_{i+1}$. Such a path must exist
  as~$v_{i+1}\in V(G_{i})\subseteq V(G_j)\subseteq
  N_r^{G_{j-1}}(v_j)$. We let~$W_{i+1}:=\bigcup_{1\leq j\leq
    i}V(P_{j,i+1})\cap N_r^{G_{i}}(v_{i+1})$. Note that~$|W_{i+1}|\leq
  i\cdot (r+1)$ (the paths have length at most~$r$ and hence consist
  of~$r+1$ vertices). It remains to be shown is that the length of any
  such play is bounded by ~${\ell}$.

  Assume towards a contradiction that Connector can play on~$G$
  for~${\ell}'={\ell}+1$ rounds. Let~$(v_1,\ldots,
  v_{{\ell}'},$ $G_1,\ldots, G_{{\ell}'}, W_1,\ldots, W_{{\ell}'})$ be
  the play. As~${\ell}'>N_\CCC(r,2s_\CCC(r))$, for~$W:=\{v_1,\ldots,
  v_{{\ell}'}\}$ there is a set~$S\subseteq V(G)$
  with~$|S|<s_\CCC(r)$, such that~$W$ contains an~$r$-scattered
  set~$I$ of size~$t:=2s_\CCC(r)$ in~$G\setminus S$.   Suppose that $I=\{u_1,\ldots,u_t\}$, where $u_j=v_{i_j}$
    for indices $1\le i_1<i_2<\ldots<i_t\le\ell'$.

  We now consider the pairs~$(u_{2j-1}, u_{2j})$ for~$1\leq j\leq
  s(r)$. By construction,~$P_j:=P_{i_{2j-1}, i_{2j}}$ is a path of
  length at most~$r$ from~$u_{2j-1}$ to~$u_{2j}$
  in~$G_{i_{2j-1}-1}$. Any path~$P_j$ must necessarily contain a
  vertex~$s_j\in S$, as otherwise the path would exist in~$G \setminus
  S$,
  contradicting the fact that~$I$ is~$r$-scattered in~$G \setminus S$. We claim
  that for~$i\neq j$,~$s_i\neq s_j$, but this is not possible, as
  there are strictly less than~$s_\CCC(r)$ vertices in~$S$. The claim
  follows easily from the following
  observation. Assume~$i>j$. Then~$V(P_j)\cap V(G_{2j-1}\subseteq W_{2j}$, thus~$V(P_j)\cap V(G_{2j+1})=\emptyset$,
  and~$V(P_i)\subseteq V(G_{2j+1})\subseteq
  V(G_{2i})$. Thus~$V(P_i)\cap V(P_j)=\emptyset$ for~$i\neq j$.
\end{proof}

\begin{remark}\label{rem:splitterrunningtime}\upshape
  In the proof of our main theorem, we will also have to compute
  Splitter's winning strategy efficiently in the following
  sense. 
  
  {\itshape Suppose that we are in a play~$v_1,W_1,\ldots,v_i,W_i$,
    and let~$G_0,G_1,\ldots,G_i$ be the graphs associated with the
    play (that is,~$G_0=G$
    and~$G_{j+1}=G_j[N_r^{G_j}(v_{j+1})\setminus W_{
      j+1}]$). For~$1\le j\le i$, let~$T_j$ be a breadth-first search
    tree of depth~$r$ in~$G_{j-1}$ with root~$v_j$.
    
    Then, given~$v_1,W_1,\ldots,v_i,W_i,v_{i+1}$ and~$T_1,\ldots,T_j$
    and Connector's move~$v_{i+1}$ in round~$(i+1)$, we can compute
    Splitter's answer~$W_{i+1}$ according to her winning strategy in
    time~$\Oof(ri|V(G_i)|+|E(G_i)|))$.  }

  To see this, recall that~$W_{i+1}:=\bigcup_{1\leq j\leq
    i}V(P_{j,i+1})\cap N_r^{G_{i}}(v_{i+1})$, where~$P_{j,i+1}$ can be
  any shortest path from~$v_j$ to~$v_{i+1}$ in~$G_{j-1}$. We choose
  the path from~$v_{i+1}$ to~$v_j$ in the tree~$T_j$. We can compute
  this path in time~$\Oof(r)$ and thus all paths in time~$\Oof(ri)$. We can
  compute~$N_r^{G_i}(v_{i+1})$ in time~$\Oof(|V(G_i)|+|E(G_i)|)$ and the
  intersection in time $\Oof(ri|V(G_i|)$.
\end{remark}

\begin{remark}\label{lem:splittersuperset}\upshape
  If Splitter wins the~$({\ell},m,r)$-splitter game on a graph~$G$, then he
  also wins if we remove in each step of the game a superset of his
  chosen set~$W$.

  We implicitly use this remark when sometimes in a graph~$G_i$
  reached after~$i$ rounds of the game and after
  choices~$v_{i+1},W_{i+1}$ in the next round we do not continue the
  game the graph~$G_{i+1}=G_i[N_r^{G_i}(v_{i+1})\setminus W_{i+1}]$,
  but in a subgraph of~$G_{i+1}$.
\end{remark}

We close the section by observing the converse of
Theorem~\ref{thm:splittergame} and hence show that the splitter game
provides another characterisation of nowhere dense classes of
graphs. 

\begin{theorem}\label{thm:splitter->nwd}
  Let~$\CCC$ be a class of graphs. If for every~$r>0$ there are~${\ell},m>0$
  such that for every graph~$G\in\CCC$, Splitter wins the~$({\ell},m,r)$-splitter
  game, then~$\CCC$ is nowhere dense.
\end{theorem}
\begin{proof}
  We show that if~$\CCC$ is not nowhere dense, i.e.,~$\CCC$ contains
  all graphs as depth-$r$ minors at some depth~$r$, then for
  all~${\ell},m>0$ there is a graph~$G\in\CCC$ such that Connector
  wins the~$({\ell},m,4r+1)$-splitter game.

  Let~${\ell},m>0$. We choose~$G\in\CCC$ such that~$G$ contains the
  complete graph~$K:=K_{{\ell}m+1}$ as a depth-$r$ minor. Connector
  uses the following strategy to win the~$({\ell},m,4r)$-splitter
  game. Connector chooses any vertex from the branch set of a vertex
  of~$K$. The~$4r+1$-neighbourhood of this vertex contains the branch
  sets of all vertices of~$K$. Splitter removes any~$m$ vertices. We
  actually allow him to remove the complete branch sets of all~$m$
  vertices he chose. In round~$2$ we may thus assume to find the
  complete graph~$K_{(\ell-1)m+1}$ as a depth-$r$ minor and continue
  to play in this way until in round~${\ell}$ at least the branch set
  of a single vertex remains.
\end{proof}

\section{Independent Sets in Nowhere Dense Classes of Graphs}
\label{sec:scattered-sets}

In this section we use the splitter game to show that the
\textsc{Distance Independent Set} problem, which is NP-complete in
general, 
is fixed-parameter tractable on nowhere dense classes of graphs. This
will be used later in the proof of our main theorem but is also of
independent interest. Recall from Section~\ref{sec:prelim} that, for
$r\geq 0$, a set
of vertices in a graph is \emph{$r$-independent} if their mutual
distance is greater than $r$. 

 \begin{theorem}\label{thm:scattered-set}
   Let~$\CCC$ be a nowhere dense class of graphs. There 
   is a function~$f$ such that for
   every~$\epsilon>0$ the following problem can be solved in
   time~$f(\epsilon,r, k)\cdot |V(G)|^{1+\epsilon}$.
   \npprob{9.5cm}{Distance Independent Set}{Graph~$G\in\CCC$,~$W\subseteq
     V(G)$,~$k, r\in\N$.}{$k+r$}{Determine whether $G$ contains an
     $r$-independent set of size $k$.}  Furthermore, if~$\CCC$ is
   effectively nowhere dense, 
   then~$f$ is computable. 
\end{theorem}

We will show that we can solve a coloured version of the problem,
called the \textsc{Rainbow Distance Independent Set} problem, and reduce
the original  distance independent set problem to the rainbow distance
independent set problem.
We first give a formal definition of rainbow sets.

\begin{definition}
  A \emph{coloured graph} $(G, C_1, \dots, C_t)$ is a graph $G$
  together with relations $C_1,\ldots C_t\subseteq V(G)$, called
  \emph{colours},  such that
  $C_i\cap C_j =\emptyset$ for all $i\not=j$.  A vertex $v\not\in
  \bigcup_{1\leq i \leq t} C_i$ is called 
  \emph{uncoloured}.
  A set $X\subseteq V(G)$ is a \emph{rainbow set}
  if all of its elements have distinct colours (and no vertex is uncoloured).

  The \textsc{Rainbow Distance Independent Set} problem on a class
  $\CCC$ of graphs is the following problem.
  \npprob{10.5cm}{
    Rainbow Distance Independent Set (Rainbow
    DIS)}{Graph~$G\in\CCC$, $C_1,\ldots, 
    C_t\subseteq V(G)$,~$k, r\in\N$.}{$k+r$}{Determine whether $G$
    contains a rainbow $r$-independent set of size $k$. }
\end{definition}

Before we describe the algorithm for solving the \textsc{Rainbow Distance
Independent Set} problem, let us show how the plain \textsc{Distance
  Independent Set} problem can be reduced to the rainbow version. 

The \emph{lexicographic product} $G\bullet H$ of two graphs $G$ and
$H$ is defined by $V(G\bullet H)= V(G)\times V(H)$ and $E(G\bullet
H)=\big\{\{(x,y), (x',y')\} : \{x,x'\}\in E(G)$ or $\big(x=x'$ and
$\{y,y'\}\in E(H)\big)\big\}$. The graph $G\bullet H$ has a
natural coloured version $G\colprod H$: we associate a colour with every vertex of
$H$ and colour every vertex of $G\bullet H$ by its projection on $H$. That is, the
colour of $(x,y)$ is $y$ (or the colour associated with $y$).
It is easy to see that a graph $G$ has an $r$-independent set of size
$k$ if and only if $G\colprod K_k$ has a rainbow $r$-independent of
size $k$. This gives us the reduction from distance independent sets
to their rainbow variant. 
Furthermore, observe that if Splitter wins the $(l, m, r)$-splitter game
on a graph $G$, for some $r, l, m\geq 0$,  then he also wins the $(l,
k\cdot m, r)$-splitter game on $G\bullet K_k$, for all $k$.
As a consequence, together with Theorem~\ref{thm:splittergame} and
Theorem~\ref{thm:splitter->nwd} this implies a different and very simple
proof of the following result by \NOdM\xspace (Theorem 13.1 of
\cite{NesetrilOdM12}) that nowhere
dense classes of graphs 
are preserved by taking lexicographic products in the following sense.

\begin{corollary}\label{cor:lexprod}
  If $\CCC$ is a nowhere dense
  class of graphs then for every $k\geq 0$, $\{ G\bullet K_k \st G\in
  \CCC\}$ is also nowhere dense.
\end{corollary}

Note, however, that the reduction above reduces \textsc{Distance
  Independent 
  Set} on a class $\CCC$ of graphs to \textsc{Rainbow
  Distance Independent Set} on the class $\bigcup_{k\ge 1}\CCC\bullet
K_k$, where $\CCC\bullet
K_k:=\{G\bullet H : G\in\CCC\}$. For the non-uniform version of our
results, this is no problem, because by the previous result, if $\CCC$ is nowhere dense then $\CCC\bullet
K_k$ is nowhere dense as well, and in the nonuniform setting we only
have to deal with fixed $k$. We need to be slightly more careful for
the uniform version. The key insight is that we can easily translate a
winning strategy for Splitter in the $(\ell,m,r)$-splitter game on a
graph $G$ to a winning strategy in the $(\ell,km,r)$-splitter game on
$G\bullet K_k$.

We are now ready to use this reduction to complete the proof of 
Theorem~\ref{thm:scattered-set}.  Let $\epsilon>0$ and let $\ell, m$ be chosen according to Theorem~\ref{thm:splittergame}
  such that Splitter has a winning strategy for the $(\ell,m,
  4k^2r)$-splitter game on every graph in $\CCC$.
 Choose
  $n_0=n_0(\epsilon)$ according to
  Theorem~\ref{thm:densitycharacterization} such that every graph
  $G\in\CCC$ of order $n\geq n_0$ has at most $n^{1+\epsilon}$ many
  edges. 

Suppose we are given an instance $G, k, r, W$ of
\textsc{Distance Independent 
  Set}, where $G\in \CCC$.  We first compute the coloured graph $G' :=
G\colprod K_k$. Let $C_1, \dots, C_t$, where $t := k$,  be the colours of $G'$. As
explained above, Splitter wins the $(l, mk, 4k^2r)$-splitter game on
$G'$ and his winning strategy can easily be computed from
any winning strategy for the $(\ell, m, 4k^2r)$-splitter game on $G$.

We need to 
decide if $(G', C_1, \dots, C_t)$ has a rainbow $r$-independent set of 
  size $k$. If
  $n=|V(G)|\leq n_0$, we test whether this set exists by brute force. In this case the
  running time is bounded by a function of $r,k$ and $\epsilon$. So
  let us assume $n\geq n_0$.

  Let $G_1:=G'$. We compute an inclusion-wise maximal rainbow $r$-independent set
  $I_1=\{x_1^1,\ldots, x_1^{k_1}\}$ of size $k_1\leq k$ by a greedy
  algorithm. If $k_1=k$, we are done and return the independent
  set. Otherwise, we may assume without loss of generality that
  $x_i^j$ has colour $j$. Let $X_1:=N_r(I_1)$. Then all elements
  with colours $k_1+1,\ldots, t$ are contained in $X_1$. Let
  $Y_1:=N_r(X_1)$. Then all paths of length at most $r$ between
  elements of colour $k_1+1,\ldots, k$ lie inside $Y_1$. Let
  $G_2:=G_1\setminus Y_1$.

  We continue by computing an inclusion-wise maximal rainbow $r$-independent set in
  $G_2$. Denote this set by $I_2=\{x_2^1,\ldots, x_2^{k_2}\}$. Note
  that all occurring colours are among $1,\ldots, k_1$ and in
  particular we have $k_2\leq k_1$ because no other colours occur in
  $G_1\setminus Y_1$. Again we may assume without loss of generality that $x_2^i$
  has colour $i$. Let $X_2:=N_r(I_2)$. Then we find all elements with
  colours $k_2+1,\ldots, t$ in $X_1\cup X_2$. We let
  $Y_2:=N_r(X_2)$. Let $G_3:=G_2\setminus Y_2$.

  We repeat this construction until $k_s=k_{s+1}$ or until
  $G_{s+1}=\emptyset$. Note that $s\leq k$, because $k_1< k$. In the
  first case we have constructed $s+1$ sets $I_i=\{x_i^1,\ldots,
  x_i^{k_i}\}$, $X_i$ and $Y_i$ such that $x_i^j$ has colour~$j$ for
  $1\leq i\leq s+1$, $1\leq j\leq k_s$. Furthermore, the colours
  $k_s+1,\ldots, t$ occur only in $X_1\cup\ldots\cup X_s$ and all
  paths of length at most $r$ between vertices of these colours lie in
  $Y_1\cup\ldots\cup Y_s$. By construction, no vertex of colour
  $k_s+1,\ldots, t$ has distance at most $r$ to any vertex of
  $I_{s+1}$. Hence we may assume that any rainbow $r$-independent
  set includes the vertices $x_{s+1}^1,\ldots, x_{s+1}^{k_s}$ of
  colour $1,\ldots, k_s$. It remains to solve the rainbow
  $r$-independent set problem with parameter $k':=k-k_s$ and colours
  $k_s+1,\ldots, t$ on $G':=G[Y_1\cup\ldots \cup Y_s]$. 

  In the other case ($G_{s+1}=\emptyset$) we also let
  $G'':=G[Y_1\cup\ldots Y_s]$. The only difference is that we have
  to solve the original problem with parameter $k'=k$.

  If $G''$ is not connected, let $U_1, \dots, U_c\subseteq G''$ be the
  components of $G''$. For all possible partitions of the set $C_1,
  \dots, C_t$ of colours into parts $\VVV_1, \dots, \VVV_c$ we
  proceed as follows. For all $1\leq i\leq c$ we delete all colours
  from $U_i$ not in $\VVV_i$, i.e.~work in the coloured graph $(U_i,
  \VVV_i)$. 
  We then solve the problem separately
  for all components $(U_i, \VVV_i)$ and for each component determine the maximal
  value $k''\leq k'$ so that $(U_i, \VVV_i)$ contains a rainbow
  $r$-independent set. We then simply check whether for some partition
  $(V_1, \dots, \VVV_c)$ of the colours the maximal values
  for the individual components sum up to at least $k'$.

  Hence, we can assume
  that $G''$ is connected. 
  Then $G'''$ has diameter at most $4k^2\cdot r$ (there are at
  most $\sum_{i=1}^k i\leq k^2$ many vertices in the independent sets
  surrounded by their $2r$-neighbourhoods of diameter at most
  $4r$). Hence the radius of $G''$ is at also at most $4k^2\cdot
  r$.

  Let $v$ be a centre vertex of $G''$. We let $v$ be Connector's choice
  in the $(\ell, km, 4k^2r)$-splitter game and let $M$ be Splitter's
  answer. Without loss of generality we assume that 
  $M=\{m_1,\ldots, m_m\}\neq \emptyset$. We let $G''':=G''\setminus M$ and continue with
  a different colouring of $G'''$ as follows. Let $X\subseteq M$ be a
  rainbow $r$-independent set in $G''$, possibly $X=\emptyset$ (we
  test for all possible sets $X\subseteq M$ whether they are rainbow
  $r$-independent sets and recurse with every possible such set).  We
  remove the colours occurring in $X$ completely from the graph and
  furthermore we remove the colour of vertices from $N_r^{G''}(X)$. 

  We now change the colours of $G'''$ as follows. For every colour
  $C_i$, with $1\leq i \leq t$, and every \emph{distance vector} $\bar{d} := (d_1,\ldots,
  d_m)$, where $d_i\in\{1,\ldots, r, \infty\}$, we add a new colour
  $C_{i,\bar{d}}$ and set $C_{i, \bar{d}}$ to be the set of all vertices
  $w\in C_i$ such that
  $\dist_{G''}(w,m_i) = d_i$, for all $1\leq i\leq r$, where we define
  $\dist_{G''}(w,m_i) = \infty$ if the distance is bigger then $r$.
  Note that the number of colours added in this way is only $t\cdot
  d'$, where $d' := (r+1)^m$ is the number of distance vectors,
  and hence only depends on the number of original colours and $r$ and
  $m$.
  We call a
  subset $C_{i_1, \bar{d}_{1}},\ldots, C_{i_{t''}, \bar{d}_{{t''}}}$ of the colours a \emph{valid
  sub-colouring} if the colours satisfy the
  following constraints:
  \begin{enumerate}
  \item If $C_{i_j, \bar{d}_j}\neq \emptyset$ for a colour which states that
    the distance to some element $m \in M$ is $r'<r$, then
    $D_{i_{j'}, \bar{d}_{j'}}=\emptyset$ for all colours which state that the
    distance to $m$ is at most $r-r'$.
  \item If $C_{i_j, \bar{d}_j}$ and
    $C_{i_{j'}, \bar{d}_{j'}}$ are colours such that $i_j = i_{j'}$
    and $\bar d_j\neq\bar d_{j'}$ then $C_{i_j, \bar{d}_j}=\emptyset$ or
    $C_{i_{j'}, \bar{d}_{j'}}=\emptyset$.
  \end{enumerate}

  We now check for all possible sub-colourings $C_{i_1,
    \bar{d}_{i_1}},\ldots, C_{i_{t''}\bar{d}_{t''}}$ of $G'''$ whether
  they are valid and for each valid sub-colouring we recursively call
  the algorithm on $G'''$ with 
  colouring $C_{i_1, \bar{d}_{i_1}},\ldots, C_{i_{t''}\bar{d}_{t''}}$
  and parameter $k'':=k'-|X|$.
 The number of valid
  sub-colourings only depends on the original number of colours and on
  $m$ and $r$.

  We claim that this procedure correctly decides whether $G''$
  contains a rainbow $r$-independent set of size $k'$. If there
  exists such a set $Z$, let $X:=M\cap Z$. Then $X$ will be considered
  as one of the potential sets to be extended by the algorithm. No
  vertex from $Z\setminus X$ may have a colour of $X$, hence we may
  remove these colours completely from the graph. Furthermore, $Z\cap
  N_r(X) = X$, hence we may remove the colours from $N_r(X)$. Also, if
  $u\in Z$ with $\dist_{G''}(u,m)=r'< r$ for some $m\in M$, then
  $v\not\in Z$ for all $v$ with $\dist_{G''}(v,m)\leq r-r'$. Hence we
  will find $Z$ in the graph where all colours which state that the
  distance to $m$ is at most $r-r'$ are removed. Conversely assume
  that the algorithm has chosen a rainbow $r$-independent set $I$ in
  $G'''$ of size $k'-|X|$ for some $X\subseteq M$ and some valid
  sub-colouring of a colouring which is consistent with $X$. By
  Condition (1) of valid sub-colourings, $I$ is also an
  $r$-independent set in $G''$. By Condition (2) of valid
  sub-colourings, $I$ is also rainbow in $G''$. 

  We now analyse the running time of the algorithm. First observe that
  in a recursive call the parameters $r$ and $m$ are left
  unchanged and $k$ can only decrease. Moreover it follows from the definition of $G'''$ that
  Splitter has a winning strategy for the $(\ell-1, km, 4k^2r)$-splitter
  game on $G'''$. Thus in each recursive call we can reduce the parameter $\ell$ by $1$. Once
  we have reached $\ell=0$, the graph $G'''$ will be empty and the
  algorithm terminates. 

  There is one more issue we need to attend to, and that is how we
  compute Splitter's winning strategy, that is, the sets~$M$. We use
  Remark~\ref{rem:splitterrunningtime}. This means that to compute~$M$
  in some recursive call, we need the whole history of the game (in a
  sense, the whole call stack). In addition, we need a breadth-first
  search tree in all graphs that appeared in the game before. It is no
  problem to compute a breadth-first search tree once when we first
  need it and then store it with the graph; this only increases the
  running time by a constant factor.

  Let us first describe the running time of the algorithm on level $j$
  of the recursion.  The time for computing $k$ maximal
  $r$-independent sets of size at most $k$ and their
  $2r$-neighbourhoods can be bounded by time $c_0\cdot
  n^{1+\epsilon}$. The factor $n^{1+\epsilon}$ stems from the
  breadth-first searches we have to perform in order to find the sets
  $Y(i)$ and Splitter's strategy and $c_0$ is a constant depending
  only on $r$, $k, \epsilon$ and $\CCC$.

  As the initial number of colours was $k$ and the number of colours
  in every recursive step increases by a factor depending only on $r$
  and $m$ (which depends only on $r,k$ and $\CCC$), the total number
  of colours depends only on $r, k$ and $\CCC$. Hence the number of
  rainbow $r$-independent subsets $X$ of an occurring set $M$ is
  bounded by a constant $c_1$ depending only on $r$, $k$ and
  $\CCC$. The number of valid sub-colourings in any recursive step is
  bounded by a constant $c_2$ depending only on $r, k$ and
  $\CCC$.

  Furthermore, for $n\leq n_0$ the running time can be bounded by a
  constant $c_3$ that only depends on $k,r, \epsilon$ and $\CCC$. For
  $j=0$, the running time can be bounded by a constant $c_4$ depending
  only on $k, r, \epsilon$ and $\CCC$. We obtain the following
  recurrence for $T$. 

  \begin{align*}
    T(0)&\le c_3+c_4,\\
    T({j})&\le c_3+c_0\cdot n^{1+\epsilon} + c_1\cdot c_2\cdot T(j-1) &\text{for
      all }{j}\ge 1.
  \end{align*}
  We conclude that there is a constant $c$ depending only on $k, r,
  \epsilon$ and $\CCC$ such that $T(\ell)\leq c\cdot n^{1+\epsilon}$.

This completes the proof of Theorem~\ref{thm:scattered-set}.
\qed

\section{Sparse Neighbourhood Covers}
\label{sec:covers}

Neighborhood covers of small radius and small size play a key role in
the design of many data structures for distributed systems.  Such
covers will also form the basis of the data structure constructed in
our first-order model-checking algorithm on nowhere dense classes of
graphs. In this section we will show that nowhere dense classes of
graphs admit sparse neighbourhood covers of small radius and small
size and present an fpt-algorithm for computing such covers.

\begin{definition}
  For~$r\in\N$, an \emph{$r$-neighbourhood cover}~$\XXX$ of a
  graph~$G$ is a set of connected subgraphs of~$G$ called
  \emph{clusters}, such that for every vertex~$v\in V(G)$ there is
  some~$X\in\XXX$ with~$N_r(v)\subseteq X$. 

  The
  \emph{radius}~$\rad(\XXX)$ of a cover~$\XXX$ is the maximum radius
  of any of its clusters. The \emph{degree}~$d^\XXX(v)$ of~$v$
  in~$\XXX$ is the number of clusters that contain~$v$. The
  \emph{maximum degree}~$\Delta(\XXX)$ of~$\XXX$
  is~$\Delta(\XXX)=\max_{v\in V(G)} d^\XXX(v)$.  The size of~$\XXX$
  is~$\norm{\XXX}=\sum_{X\in\XXX}|X|=\sum_{v\in V(G)}d^\XXX(v)$.
\end{definition}

The main result of this section is the following theorem.

\begin{theorem}\label{thm:alg-covers}
  Let~$\CCC$ be a nowhere dense class of graphs. There is a
  function~$f$ such that for all~$r\in\N$ and~$\epsilon>0$ and all
  graphs~$G\in\CCC$ with~$n\geq f(r,\epsilon)$ vertices, there exists
  an~$r$-neighbourhood cover of radius at most~$2r$ and maximum degree
  at most~$n^\epsilon$ and this cover can be computed in
  time~$f(r,\epsilon)\cdot n^{1+\epsilon}$. Furthermore, if~$\CCC$ is
  effectively nowhere dense, then~$f$ is computable.
\end{theorem}

To prove the theorem we use the concept of generalised colouring
numbers introduced by Kierstead and Yang in \cite{KiersteadY03}.  For
a graph~$G$, let~$\Pi(G)$ be the set of all linear orderings
of~$V(G)$. 
For~$u,v\in V(G)$ and~$k\in \N$, we say that~$u$ is
\emph{weakly~$k$-accessible} from~$v$ with respect to~$<\in\Pi(G)$
if~$u< v$ and there is a~$u{-}v$-path~$P$ of length at most~$k$ such
that for all~$w\in V(P)$ we have~$u\leq w$. We write~$\leq$ for the
reflexive ordering induced by~$<$. Let~$\wreach_k(G,<, v)$ be the set
of vertices that are weakly $k$-accessible from~$v$ and let
$\wreach_k[G,<,v]:=\wreach_k(G,<,v)\cup\{v\}$. The
\emph{weak~$k$-colouring number}~$\wcol_k(G)$ of~$G$ is defined as
\begin{align*}
  \wcol_k(G)=\min_{<\in\Pi(G)}\max_{v\in V(G)}|\wreach_k[G,<,v]|.
\end{align*}

Zhu \cite{Zhu09} (and in fact also Kierstead and Yang but they were
not aware of the depth-$r$ minor terminology) showed that general
colouring numbers and densities of depth-$r$ minors are strongly
related. From this, Ne{\v s}et{\v r}il and Ossona de Mendez conclude
that the weak colouring number on nowhere dense classes is small.

\begin{lemma}[\cite{Zhu09,  NesetrilO11}]\label{thm:wcolcharacterization}
  Let~$\CCC$ be a nowhere dense class of graphs. Then there is a
  function~$f$ such that for every~$r\in\N$, every~$\epsilon>0$, every
  graph~$G\in\CCC$ with~$n\geq f(r,\epsilon)$ vertices satisfies
  $\wcol_r(G)\leq n^\epsilon$. Furthermore, if~$\CCC$ is effectively
  nowhere dense, then~$f$ is computable.
\end{lemma}

For our purpose, we need an efficient algorithm for ordering the
vertices of~$G$ in an order witnessing $\wcol_r(G)\leq
n^\epsilon$. Dvo\v r\'ak \cite{Dvorak13} conjectures that in general
computing~$\wcol_r(G)$ is NP-complete. We are able to prove his
conjecture for all~$r\geq 3$. He provides an approximation algorithm
to solve the problem, but its running time is~$\Oof(r\cdot n^3)$ which
is too expensive for our purpose.  We propose a more efficient
approximation algorithm, based on Ne{\v s}et{\v r}il and Ossona de
Mendez's transitive fraternal augmentation technique and an argument
from Zhu's proof.

In the following we will work with ordered representations of graphs
where each vertex stores an adjacency list for its in-neighbours and
an adjacency list for its out-neighbours.

\begin{definition}
  Let~$\vec{G}$ be a directed graph. A \emph{tight~$1$-transitive
    fraternal augmentation} of~$\vec{G}$ is a directed graph~$\vec{H}$
  on the same vertex set such that for all distinct vertices~$u,v,w$
  \begin{itemize}
  \item if~$(u,v)\in E(\vec{G})$, then~$(u,v)\in E(\vec{H})$. 
  \item if~$(u,w), (w,v)\in E(\vec{G})$, then~$(u,v)\in E(\vec{H})$,
  \item if~$(u,w), (v,w)\in E(\vec{G})$, then~$(u,v)$ or~$(v,u)$ are
    arcs of~$\vec{H}$ and
  \item for all~$(u,v)\in E(\vec H)$, either~$(u,v)\in E(\vec G)$ or
    there is some~$w$ such that~$(u,w),(w,v)\in E(\vec G)$ or
   ~$(u,w),(v,w)\in E(\vec G)$.
  \end{itemize}
  We write~$\aug(\vec{G},1)$ for any tight~$1$-transitive fraternal
  augmentation of~$\vec{G}$ and for~$r>1$ we write~$\aug(\vec{G}, r)$
  for~$\aug(\aug(\vec{G}, r-1),1)$. We call~$\aug(\vec{G},r)$ a
  tight~$r$-transitive fraternal augmentation of~$\vec{G}$. We will
  often write $\aug(G,r)$ and speak of an $r$-transitive fraternal
  augmentation of $G$ instead of $\aug(\vec{G}, r)$ and an
  $r$-transitive fraternal augmentation of an orientation $\vec{G}$ of
  $G$.
\end{definition}

In \cite{NesetrilO05}, Ne{\v s}et{\v r}il-Ossona de Mendez show how to
efficiently compute tight transitive fraternal augmentations. They
state the result in terms of average densities of depth-$r$ minors,
for our purpose it suffices to state their result for nowhere dense
classes. All functions~$f(r, \epsilon)$ in the following lemmas are
computable if~$\CCC$ is effectively nowhere dense.

\begin{lemma}[Ne{\v s}et{\v r}il-Ossona de Mendez \cite{NesetrilO05},
  Corollary 4.2, Theorem 4.3] \label{thm:computingorientation} 

  Let~$\CCC$ be a nowhere dense class of graphs. There is a
  function~$f$ such that for all~$r\in\N$ and~$\epsilon>0$ and all
  graphs~$G\in\CCC$ with~$n\geq f(r,\epsilon)$ vertices, there exists
  an~$r$-transitive fraternal augmentation~$\vec{H}=\aug(G,r)$ of~$G$
  such that~$\Delta^-(\vec{H})\leq n^\epsilon$. Furthermore,~$\vec{H}$
  can be computed from~$G$ in time~$f(r,\epsilon)\cdot
  n^{1+\epsilon}$.
\end{lemma}

We will write~$\aug(G,r,\epsilon)$ for an augmentation
$\vec{H}=\aug(G,r)$ such that~$\Delta^-(\vec{H})\leq n^\epsilon$.

The following property of transitive fraternal augmentations is noted
in the proof of Theorem 5.1 in \cite{NesetrilO05}.

\begin{lemma}[\cite{NesetrilO05}]\label{thm:augmentationneighbourhood}
  Let~$G$ be a graph and let~$r\in\N$. Let~$\vec{H}=\aug(G,r)$ be
  an~$r$-transitive fraternal augmentation of~$G$. Let~$v\in V(G)$
  and~$w\in N_r^G(v)$. Let~$v=v_1, v_2,\ldots, v_l=w$ be a path of
  length at most~$r$ from~$v$ to~$w$ in~$G$. Then either~$(v,w)\in
  E(\vec{H})$ or~$(w,v)\in E(\vec{H})$ or there is some~$v_i$ such
  that~$(v_i,v), (v_i,w) \in E(\vec{H})$.
\end{lemma}

In fact, for the results in the previous lemma it would suffice to
use an~$\left\lceil \log_{3/2} r\right\rceil+1$-augmentation. While
this would make the construction more efficient, we refrain from doing
so for ease of presentation.

We now show how to approximate~$\wcol_r(G)$ with the help
of~$r$-transitive fraternal augmentations.

\begin{lemma}\label{thm:approximatingcol}
  Let~$G$ be a graph and let~$r>0$. Let~$\vec{H}=\aug(G,r)$ be
  an~$r$-transitive fraternal augmentation of~$G$ such
  that~$\Delta^-(\vec{H})\leq d$. Then~$\wcol_r(G)\leq 2(d+1)^2$.
\end{lemma}

\begin{proof}
  As~$\Delta^-(\vec{H})\leq d$, the underlying undirected graph~$H$
  is~$2d$-degenerate and we can order the vertices of~$H$ such that
  each vertex has at most~$2d$ smaller neighbours. Denote this order
  by~$<$. For each vertex~$v\in V(G)$ we count the number of
  end-vertices of paths of length at most~$r$ from~$v$ such that the
  end-vertex is the smallest vertex of the path. This number
  bounds~$|\wreach_r[G,<, v)]|$.

  By Lemma~\ref{thm:augmentationneighbourhood}, for each such path
  with end-vertex~$w$, we either have an edge~$(v,w)$ or an
  edge~$(w,v)$ or there is~$u$ on the path and we have edges~$(u,v),
  (u,w)$ in~$H$. By construction of the order there are at most~$2d$
  edges~$(v,w)$ or~$(w,v)$ such that~$w<v$. Furthermore, we have at
  most~$d$ edges~$(u,v)$, as~$v$ has indegree at most~$d$ and for each
  such~$u$ there are at most~$2d$ edges~$(u,w)$ such that~$w< u$ by
  construction of the order. These are exactly the pairs of edges we
  have to consider, as no vertex on the path from~$v$ to~$w$ may be
  smaller than~$w$.  Hence in total we have~$|\wreach_r[G,<,v]|\leq
  2d+2d^2+1\leq 2(d+1)^2$.
\end{proof}

\begin{corollary}\label{thm:order}
  Let~$\CCC$ be a nowhere dense class of graphs. There is a
  function~$f$ such that for all~$r\in\N$ and~$\epsilon>0$ and
  every~$G\in\CCC$ with~$n\geq f(r,\epsilon)$ vertices, we can order
  the vertices of~$G$ in order~$<$ such that~$|\wreach_r[G,<, v]|\leq
  n^\epsilon$ for all~$v\in V(G)$ in time~$f(r,\epsilon)\cdot
  n^{1+\epsilon}$. Furthermore, if~$\CCC$ is effectively nowhere
  dense, then~$f$ is computable.
\end{corollary}

\begin{proof}
  Let~$\delta:=\epsilon/4$. We compute an~$r$-transitive fraternal
  augmentation~$\vec{H}=\aug(G,r, \delta)$ of~$G$ in
  time~$g(r,\delta)\cdot n^{1+\delta}$ by
  Lemma~\ref{thm:computingorientation}, where~$g$ is the function from
  the lemma. We can order the vertices as in the proof of
  Theorem~\ref{thm:approximatingcol} by a simple greedy algorithm in
  time~$\Oof(n^{1+\delta})$ and obtain an order
  witnessing~$\wcol_r(G)\leq 2(n^{\delta}+1)^2\leq n^\epsilon$.
\end{proof}

In the next lemma we use the weak colouring number to prove the
existence of sparse neighbourhood covers in nowhere dense classes of
graphs.  

\begin{definition}
  Let~$G$ be a graph, let~$<$ be an ordering of~$V(G)$ and
  let~$r>0$. For a vertex~$v\in V(G)$ we define
  \[X_{r}[G,<,v] := \{w\in V(G) : v\in \wreach_{r}[G,<,w]\}.\]
\end{definition}

\begin{lemma}\label{thm:covers}
  Let~$G$ be a graph such that~$\mathrm{wcol}_{2r}(G)\leq s$ and
  let~$<$ be an order witnessing this. Then~$\XXX=\{X_{2r}[G,<,v] :
  v\in V(G)\}$ is an~$r$-neighbourhood cover of~$G$ with radius at
  most~$2r$ and maximum degree at most~$s$.
\end{lemma}
\begin{proof}
  Clearly the radius of each cluster is at most~$2r$, because if~$v$
  is weakly~$2r$-accessible from~$w$ then~$w\in
  N_{2r}(v)$. Furthermore, every~$r$-neighbourhood lies in some
  cluster. To see this, let~$v\in V(G)$. Let~$u$ be the minimum
  of~$N_r(v)$ with respect to~$<$. Then~$u$ is weakly~$2r$-accessible
  from every~$w\in N_r(v)\setminus\{u\}$ as there is a path from~$w$
  to~$u$ which uses only vertices of~$N_r(v)$ and has length at
  most~$2r$ and~$u$ is the minimum element
  of~$N_r(v)$. Thus~$N_r(v)\subseteq X_{2r}[G,<,u]$. Finally observe
  that for every~$v\in V(G)$,
 \begin{align*}
   d^\mathcal{X}(v) &= |\{u\in V(G) : v\in X_{2r}[G,<,u]\}|\\ &=
   |\{u\in V(G) : u\in \wreach_{2r}[G_<, v]\}| = |\wreach_{2r}[G_<,v]|\leq s.
 \end{align*}
\end{proof}

\begin{proof}[Proof of Theorem~\ref{thm:alg-covers}]
  Let~$\delta:=\epsilon/2$. We order the vertices of~$G$ in order~$<$
  as in Corollary~\ref{thm:order}, where~$\delta$ plays the role
  of~$\epsilon$ in the corollary, such that~$\wreach_{2r}[G,<,v]\leq
  n^\delta$ for all~$v\in V(G)$ in time~$g(r,\delta)\cdot
  n^{1+\delta}$, where~$g$ is the function from the corollary.

  Let us first note the following observation. 

  \begin{Claim}\label{claim:correctness}
    For~$v\in V(G)$ let~$S(v):=\{u : u<v\}$. Then~$X_{2r}[G,<,v] =
    N_{2r}^{G\setminus S(v)}(v)$.
  \end{Claim}

  Our algorithm computes the sets~$X_{2r}[G,<,v]$ in ascending
  order. To do so, it chooses the smallest vertex~$v$, performs~$2r$
  levels of a breadth-first search and then deletes~$v$ from the
  graph.  Correctness of the algorithm follows immediately from
  Claim~\ref{claim:correctness}. Let us analyse the running time.

  We construct the following representation of~$G$ which is easily
  seen to be computable in time~$\Oof(n^{1+\delta})$ . We split the
  edges of~$G$ into edges going to larger elements and into edges
  going to smaller elements with respect to the ordering. For
  each~$v\in V(G)$ we write~$N_>(v)$ (resp.~$N_<(v)$) for the
  neighbours of~$v$ that are larger (resp. smaller) than~$v$. We
  write~$d_>(v)$ for~$|N_>(v)|$ and~$d_<(v)$ for~$|N_<(v)|$. Note that
  we have~$d_<(v)\leq n^\delta$ for each~$v\in V(G)$, as~$d_<(v)\leq
  |\wreach_{2r}[G,<,v]|$.

  Let~$G'$ be a subgraph of~$G$ with~$n'$ vertices. We can count the
  edges of~$G'$ by counting the sum of~$d_<(v)$ over all~$v\in V(G')$,
  hence~$G'$ has at most~$n'\cdot n^\delta$ many edges. We can thus
  perform each breadth-first search to compute~$X_{2r}[G,<,v]$ in
  time~$\Oof(|X_{2r}[G,<,v]|\cdot n^\delta)$ for each vertex~$v\in
  V(G)$.  Furthermore, we have the following overhead in the
  breadth-first search for deleting edges that point to~$v$, which
  must be deleted. As we store the edges of each vertex in separate
  lists, for each vertex~$w\in N_>(v)$ (this is the first level of the
  breadth-first search), we have to access only the edges to vertices
  of~$N_<(w)$. No other vertex is connected to~$v$ in~$G\setminus S(v)$. Hence,
  the deletion of~$v$ from the adjacency list of~$w$ can be done in
  time~$d_<(w)\leq n^\delta$. The number of such vertices~$w$
  is~$d_>(v)$, which at the time of deletion of~$v$ is bounded
  by~$|X_{2r}[G,<,v]|$.

  For ease of presentation let~$X_v:=X_{2r}[G,<,v]$ and let us drop
  any constant factors in the following estimation. We get a total
  running time of
  \begin{align*}
    & \sum_{v\in V(G)} \big(|X_v|\cdot n^\delta+\sum_{w\in
      N_>(v)} d_<(w)\big)\\
    = & \sum_{v\in V(G)} |X_v|\cdot n^\delta+\sum_{v\in
      V(G)}\sum_{w\in
      N_>(v)} d_<(w)\\
    \leq & \sum_{v\in V(G)}
    |X_v|\cdot n^\delta+\sum_{v\in V(G)}|X_v|\cdot n^\delta\\
    = & \; 2n^\delta\sum_{v\in V(G)} |X_v|\\
    \leq & \; 2n^{1+2\delta}=:f(r,\epsilon)\cdot
    n^{1+\epsilon}
  \end{align*}
\end{proof}

\begin{remark}
  By definition, an~$r$-neighbourhood cover~$\XXX$ of a graph~$G$
  contains for each~$v\in V(G)$ a cluster~$X\in \XXX$ such
  that~$N_r^G(v) \subseteq X$. For the algorithmic applications below
  it will be useful to store along with the neighbourhood cover a
  function~$f_\XXX \st V(G) \rightarrow \XXX$ which associates with
  every vertex~$v$ such a cluster~$X$ containing
  its~$r$-neighbourhood.

  The proof of the previous theorem can easily be modified to compute
  such a function along with the neighbourhood cover as follows: we
  associate with~$v\in V(G)$ the set~$X_{2r}[G, <, u]$ for
  the~$<$-minimal~$u\in V(G)$ such that~$v\in N_r^{G\setminus S(u)}(u)$,
  where~$S(u)$ is defined as in Claim~\ref{claim:correctness} in the
  proof of Theorem~\ref{thm:alg-covers}. As the sets~$X_{2r}[G, <, u]$
  are computed in increasing order, this can be done at no extra cost.
\end{remark}

We remark that our construction also yields very good covers for other
restricted classes of graphs, in particular for classes with excluded
minors and classes of graphs of bounded expansion, where we can
replace the maximum degree $n^\epsilon$ of the neighbourhood cover by
a constant. See the conclusions (Section~\ref{sec:conclusion}) for
further comments.


\section{Locality of First-Order Logic}\label{sec:local}

In this chapter, we prove the ``rank-preserving'' version of
Gaifman's locality theorem stated in the introduction.

\subsection{Background on First-Order Logic}
We start with a brief review of first-order logic. For background, we
refer the reader to \cite{EbbinghausFluTho94}. A \emph{(relational) vocabulary} is
a finite set of relation symbols, each with a prescribed
arity. Throughout this paper, we let~$\sigma$ be a vocabulary.  A
\emph{$\sigma$-structure} $A$ consist of a (not necessarily finite) set~$V(A)$, called the
\emph{universe} or \emph{vertex set} of~$A$, and for each~$k$-ary
relation symbol~$R\in\sigma$ a~$k$-ary relation~$R(A)\subseteq
V(A)^k$. A structure~$A$ is \emph{finite} if its universe is. 

For example, graphs may be viewed as~$\{E\}$-structures, where~$E$ is
a binary relation symbol.

Let~$A$ be a~$\sigma$-structure. For a subset~$X\subseteq V(A)$, the
\emph{induced substructure} of~$A$ with universe~$X$ is the
$\sigma$-structure~$A[X]$ with~$V(A[X])=X$ and~$R(A[X])=R(A)\cap X^k$
for every~$k$-ary~$R\in\sigma$. For a vocabulary
$\sigma'\subseteq\sigma$, the \emph{$\sigma'$-restriction} of~$A$ is
the~$\sigma'$-structure~$A'$ with~$V(A')=V(A)$ and~$R(A')=R(A)$ for
all~$R\in\sigma'$. Conversely,~$A$ is a \emph{$\sigma$-expansion} of a
$\sigma'$-structure~$A'$ if~$A'$ is the~$\sigma'$-restriction of~$A$.

\emph{First-order formulas} of vocabulary~$\sigma$ are formed from
atomic formulas~$x=y$ and $R(x_1,\ldots,x_k)$, where~$R\in\sigma$ is
a~$k$-ary relation symbol and~$x,y,x_1,\ldots,x_k$ are variables (we
assume that we have an infinite supply of variables) by the usual
Boolean connectives~$\neg$~(negation),~$\wedge$ (conjunction),
and~$\vee$ (disjunction) and existential and universal
quantification~$\exists x,\forall x$, respectively. The set of all
first-order formulas of vocabulary $\sigma$ is denoted
by~$\FO[\sigma]$, and the set of all first-order formulas
by~$\FO$. The free variables of a formula are those not in the scope
of a quantifier, and we write~$\phi(x_1,\ldots,x_k)$ to indicate that
the free variables of the formula~$\phi$ are among $x_1,\ldots,x_k$. A
\emph{sentence} is a formula without free variables. The
\emph{quantifier rank}~$\qr(\phi)$ of a formula~$\phi$ is the nesting
depth of quantifiers in~$\phi$, defined recursively in the obvious
way. A formula without any quantifiers is called
\emph{quantifier-free}.

To define the semantics, we inductively define a
satisfaction relation~$\models$, where for a~$\sigma$-structure~$A$, a
formula~$\phi(x_1,\ldots,x_k)$, and elements~$a_1,\ldots,a_k\in V(A)$,
\[
A\models\phi(a_1,\ldots,a_k)
\]
means that~$A$ satisfies~$\phi$ if the free variables~$x_1,\ldots,x_k$
are interpreted by~$a_1,\ldots,a_k$, respectively. If
$\phi(x_1,\ldots,x_k)=R(x_1,\ldots,x_k)$ is atomic, then
$A\models\phi(a_1,\ldots,a_k)$ if~$(a_1,\ldots,a_k)\in R(A)$. The
meaning of the equality symbol, the Boolean connectives, and the
quantifiers is the usual one.

For example, consider the formula~$\phi(x_1,x_2)=\forall y(x_1=y\vee
x_2=y\vee E(x_1,y)\vee E(x_2,y))$ in the vocabulary~$\{E\}$ of
graphs. For every graph~$G$ and vertices~$v_1,v_2\in V(G)$ we
have~$G\models\phi(v_1,v_2)$ if any only if~$\{v_1,v_2\}$ is a
dominating set of~$G$. Thus~$G$ satisfies the sentence~$\exists
x_1\exists x_2\phi(x_1,x_2)$ if, and only if, it has a
(nonempty)
dominating set of size at most~$2$.

Whenever a~$\sigma$-structure occurs as the input of an algorithm, we
implicitly assume that it is finite and encoded in a suitable
way. Similarly, we assume that formulas~$\phi$ appearing as input are
encoded suitably. By~$|\phi|$, we denote the length of the encoding of
$\phi$.

A formula~$\phi(x_1,\ldots,x_k)\in\FO[\sigma]$ is \emph{valid} if for all
$\sigma$-structures~$A$ and all elements~$a_1,\ldots,a_k\in V(A)$
it holds that~$A\models\phi(a_1,\ldots,a_k)$. 
The Completeness Theorem for First-Order Logic implies that the set of
valid formulas is recursively enumerable. 
Two formulas~$\phi(x_1,\ldots,x_k),\psi(x_1,\ldots,x_k)\in\FO[\sigma]$ are \emph{equivalent} if for all
$\sigma$-structures~$A$ and all elements~$a_1,\ldots,a_k\in V(A)$ we have~$A\models\phi(a_1,\ldots,a_k)\iff A\models\psi(a_1,\ldots,a_k)$. 

Up to logical equivalence, for all~$k,q$ there are only finitely many
$\FO$-formulas~$\phi(x_1,\ldots,x_k)$ of quantifier-rank at most
$q$. Indeed, by systematically renaming the bound variables, bringing Boolean
combinations into conjunctive normal form, and deleting duplicate
entries from the disjunctions and conjunctions, we can normalise
$\FO$-formulas in such a way that every formula can be effectively
translated into an equivalent normalised formula of the same
quantifier rank, and for all~$k,q$ the set~$\Phi(\sigma,k,q)$ of all
normalised~$\FO$-formulas~$\phi(x_1,\ldots,x_k)$ of quantifier rank at
most~$q$ is finite and computable.

The \emph{Gaifman graph}~$G_A$ of a~$\sigma$-structure~$A$ is the
graph with vertex set~$V(A)$ and an edge between
$a_1,a_2\in V(A)$ if~$a_1,a_2$ appear together in some tuple of some
relation in~$A$. The \emph{distance}~$\dist^A(a,b)$, or just~$\dist(a,b)$, between two
elements~$a,b\in V(A)$ in~$A$ is the length of the shortest path from
$a$ to~$b$ in~$G_A$, and the \emph{$r$-neighbourhood} of~$a$ in~$A$ is
the set~$N_r^A(a)$, or just~$N_r(a)$, of all~$b\in V(A)$ such that
$\dist(a,b)\le r$. For a tuple~$\tup a=(a_1,\ldots,a_k)$, we let
$N_r(\tup a)=\bigcup_{i=1}^kN_r(a_i)$.

A first-order formula~$\psi(\tup x)$ is called \emph{$r$-local} if its
truth value at a  tuple~$\tup a$ of vertices in a structure~$A$ only depends on the
$r$-neighbourhood of~$\tup a$ in~$A$, that is,~$A\models\phi(\tup
a)\iff A[N_r(\tup a)]\models\phi(\tup a)$.  For all~$d\ge 0$ there is an
$\FO$-formula~$\delta_{\le d}(x,y)$ stating that the distance between~$x$ and~$y$ is
at most~$d$. We write~$\delta_{>d}(x,y)$ instead of~$\neg\delta_{\le
  d}(x,y)$. A \emph{basic
  local sentence} is a first-order sentence of the form 
\begin{equation}
  \label{eq:basic-local}
  \exists x_1\ldots\exists
x_k\big(\bigwedge_{1\le i<j\le
  k}\delta_{>2r}(x_i,x_j)\wedge\bigwedge_{i=1}^k\phi(x_i)\big),
\end{equation}
where
$\phi$ is~$r$-local.  

\begin{theorem}[Gaifman's Locality Theorem~\cite{Gaifman82}]
  Every first-order sentence is equivalent to a Boolean combination of
  basic local sentences.
\end{theorem}

The algorithm of Frick and Grohe~\cite{frigro01}
for deciding first-order properties on graph classes of bounded local
tree width relies on Gaifman's theorem. Unfortunately, we cannot use
Gaifman's theorem here, at least not directly, because it does not give us sufficient
control over the quantifier rank of the basic local sentences we
translate a sentence to. As we intend to apply the theorem repeatedly,
such control will be crucial. To get around these difficulties, we
need a discounted rank measure, which does not charge the full
quantifier rank to distance formulas,  and a refined version of
Gaifman's theorem.

\subsection{The Logic~$\FO^+$}
We define an extension~$\FO^+$ of first-order logic by adding new
atomic formulas~$\dist(x,y)\le d$, for all variables~$x,y$ and all
$d\in\mathbb N$. We call these formulas \emph{distance atoms}. The
meaning of the distance atoms is
obvious. Note that every~$\FO^+$-formula~$\phi$ is equivalent to an
$\FO$-formula~$\phi^-$ obtained from~$\phi$ by replacing each distance
atom~$\dist(x,y)\le d$ by the
$\FO$-formula~$\delta_{\le d}(x,y)$. Thus~$\FO^+$ is only a syntactic
extension of~$\FO$. However, the quantifier rank of~$\delta_{\le
  d}(x,y)\in\FO$ is at least~$\lceil\log d\rceil$, whereas by
definition the
quantifier rank of the atomic~$\FO^+$-formula~$\dist(x,y)\le d$ is~$0$.
With this definition as one of the base steps, we can define the
quantifier rank~$\qr(\phi)$ for~$\FO^+$-formulas~$\phi$ recursively
as for~$\FO$-formulas.

We now define the discounted rank measure. Let~$q\in\N$. 

We say that~$\phi$ has \emph{$q$-rank} at most~$\ell$ if~$\phi$ has
quantifier-rank at most~$\ell$ and if each distance atom
$\dist(x,y)\leq d$ in the scope of~$i\leq \ell$ quantifiers satisfies
$d\leq (4q)^{q+\ell-i}$.

For example, the sentence 
\[
\exists x\exists y\Big(\dist(x,y)\le 12^5\wedge\exists
z\big(\dist(x,z)\le 12^6\wedge\forall
z'(\neg\dist(z,z')\le 12^4\vee\dist(z',y)\le 12^4)\big)\Big)
\]
has~$3$-rank~$6$, because for the distance atom~$\dist(x,z)\le 12^6$
in the scope of~$3$ quantifiers we have~$12^6=(4\cdot
3)^{3+6-3}$. Note that the quantifier-rank of this formula is~$4$ and
hence~$\leq\ell=6$.

For convenience, we let 

\begin{equation}
  f_q(\ell):=(4q)^{q+\ell}.\label{eq:f}
\end{equation}
This is is the
largest value of~$d$ which may occur in a distance atom~$\dist(x,y)\le
d$ of a formula of~$q$-rank~$\ell$.

The definition of the~$q$-rank
arises from the necessities of the proof of
Theorem~\ref{thm:rplocality}. Note that this rank measure makes it
cheaper to define distances as in \FO-formulas: with an~$\FO^+$-formula of~$q$-rank~$q$
we can define distances up to~$(4q)^{2q}$, which is much more than the
distance~$2^q$ we can define with an~$\FO$-formula of quantifier rank
$q$. Also note that defining distances becomes more expensive in the
scope of quantifiers.

Up to logical equivalence, for all~$k,q,\ell$ there are only finitely
many~$\FO^+[\sigma]$-formulas~$\phi(x_1,\ldots,x_k)$ of~$q$-rank at most~$\ell$. As~$\FO$-formulas, we can normalise~$\FO^+$
formulas such that
every formula can be effectively translated into an equivalent
normalised formula
of the same rank, and for all~$k,q,\ell$ the set
$\Phi^+(\sigma,k,q,\ell)$ of
all normalised~$\FO^+$-formulas~$\phi(x_1,\ldots,x_k)$ of~$q$-rank at
most~$\ell$
is finite and
computable.

\subsection{An Ehrenfeucht-Fra\"iss\'e Game for~$\FO^+$}

For~$\sigma$-structures~$A,B$ and tuples~$\tup a=(a_1,\ldots,a_k)\in
V(A)^k,\tup b=(b_1,\ldots,b_k)\in V(B)^k$ we write~$(A,\tup
a)\equiv^+_{q,\ell} (B,\tup b)$ (and say that~$(A,\tup a)$ and
$(B,\tup b)$ are
$(q,\ell)^+$-equivalent) if for all~$\phi(\tup x)\in\FO^+$ of~$q$-rank
at most~$\ell$ we have~$A\models\phi(\tup a)\iff B\models\phi(\tup
b)$. Observe that~$(A,\tup a)\equiv^+_{q,\ell} (B,\tup b)$ implies for
all~$i,j\in[k]$ that either~$\dist(a_i,a_j)=\dist(b_i,b_j)$ or
$\dist(a_i,a_j)>f_q(\ell)$ and~$\dist(b_i,b_j)>f_q(\ell)$.

We generalise the well-known characterisation of first-order
equivalence by means of the 
Ehrenfeucht-Fra\"iss\'e (EF) game (see, for example,
\cite{EbbinghausFluTho94}) to the logic~$\FO^+$ parameterized by
$q$-ranks. A \emph{partial~$d$-isomorphism} between two structures
$A,B$ is a mapping~$p$ with domain~$\dom(p)\subseteq V(A)$ and range
$\rg(p)\subseteq V(B)$ that is an isomorphism between the induced
substructure~$A[\dom(p)]$ and the induced substructure~$B[\rg(p)]$ and
in addition, preserves distances up to~$d$, that is, for all
$a,a'\in\dom(p)$ either~$\dist(a,a')=\dist(p(a),p(a'))$ or
$\dist(a,a')>d$ and~$\dist(p(a),p(a'))>d$. 

\begin{definition}[EF$_q^+$-game]\upshape
  Let~$A,B$ be~$\sigma$-structures,~$\tup{a}=(a_1,\ldots,a_k)\in V(A)^k$,~$\tup{b}=(b_1,\ldots,b_k)\in
  V(B)^k$ and~$q\in\N$. Let~$0\leq \ell\leq q$. The \emph{$\ell$-round
    EF$_q^+$-game} on~$(A, \tup a, B,\tup b)$ is played by
  two players, called \emph{Spoiler} and \emph{Duplicator}. The game
  is played for~$\ell$ rounds.  In round~$i$, Spoiler picks an element
 ~$a_{k+i}\in V(A)$ or an element~$b_{k+i}\in V(B)$. If Spoiler picks~$a_{k+i}\in
  V(A)$, then Duplicator must choose an element~$b_{k+i}\in V(B)$ and if
  Spoiler picks~$b_{k+i}\in V(B)$, then Duplicator must choose an element
 ~$a_{k+i}\in V(A)$. Duplicator wins
  the game if for~$0\leq i\leq \ell$, the mapping 
 ~$a_j\mapsto b_j$ for~$1\le j\le k+i$ is a partial
 ~$f_q(\ell-i)$-isomorphism.
\end{definition}

\pagebreak
\begin{theorem}\label{thm:game}
  For all~$q, 0\leq \ell\leq q, A, B$ and~$\tup{a}\in V(A)^{k}, \tup{b}\in
  V(B)^{k}$, the following are equivalent.
  \begin{enumerate}
  \item \label{item:ws}Duplicator has a winning strategy for the
   ~$\ell$-round EF$_q^+$ game on~$(A, \tup a,B,\tup
    b)$.
  \item \label{item:equiv}$(A,\tup a)\equiv^+_{q, \ell} (B,\tup b)$.
  \end{enumerate}
\end{theorem}

The proof of Theorem~\ref{thm:game} requires some familiarity with
logic. It is similar to the proof that equivalence in first-order
logic is characterised by the standard Ehrenfeucht-Fra\"iss\'e game
(see, for example, \cite{EbbinghausFluTho94}).

For~$\tup{a}=(a_1,\ldots,a_k)\in V(A)^k$ and~$a\in V(A)$, write
$\dist(\tup{a},a)=_{q,\ell}\tup{d}\in (\{0,\ldots, f_q(\ell)\}\cup\{\infty\})^k$ if for
all~$i\in[k]$ we have~$\dist(a_i, a)=d_i\leq f_q(\ell)$ or
$\dist(a_i,a)>f_q(\ell)$ and $d_i=\infty$. 
Note that we can easily
write a quantifier-free~$\FO^+$-formula of~$q$-rank
$\ell$ expressing~$\dist(\tup{x},x)=_{q,\ell}\tup{d}$.

We can rephrase the existence of a winning strategy for Duplicator
in the~$\ell$-round EF$_q^+$ game on~$(A, \tup a,B,\tup
    b)$ as follows. 
    \begin{itemize}
    \item 
Duplicator has a
winning strategy for the
$0$-round~$\FO_q^+$-game on~$(A, \tup a, B,\tup b)$ if, and
only if,~$\tup{a}\mapsto\tup{b}$ is a partial~$f_q(0)$-isomorphism. 

\item
For
$0<\ell\leq q$, Duplicator has a winning strategy  for the~$\ell$-round~$\FO_q^+$-game on~$(A, \tup a, B,\tup b)$ if, and only if,
\begin{enumerate}
\item[(1)] $\tup{a}$ and~$\tup{b}$ satisfy the same distance formulas up to
 ~$f_q(\ell)$ and
\item[(2)] for every~$a\in V(A)$ 
  there is a~$b\in V(B)$ such that
  Duplicator has a winning strategy for the~$\ell-1$-round~$\FO_q^+$-game on~$(A, \tup{a}a)$
  and~$(B,\tup{b}b)$ and
\item[(3)] for every~$b\in V(B)$ 
  there is an~$a\in V(A)$ such that
  Duplicator has a winning strategy for the~$\ell-1$-round~$\FO_q^+$-game on~$(A, \tup{a}a)$
  and~$(B,\tup{b}b)$.
\end{enumerate}
    \end{itemize}

This description of winning strategies can be defined in~$\FO^+$ as
follows. Let~$A$ and
$q\in \N$ be given. For~$\tup{a}=(a_1,\ldots, a_k)\in V(A)^k$, 
$\tup{x}:=(x_1,\ldots, x_k)$ and~$0\leq \ell\leq q$, let

  \begin{align*}
    \vartheta_{\tup{a}}^{q,{\ell}}(\tup x) := \bigwedge_{\substack{a_i, a_j\in \tup{a}\\\dist(a_i, a_j)=d\leq f_q({\ell})}}\dist(x_i, x_j)=d
    \quad\wedge \bigwedge_{\substack{a_i, a_j\in \tup{a}\\\dist(a_i, a_j)> f_q({\ell})}}\dist(x_i, x_j)>f_q({\ell}).
  \end{align*}

  For~${\ell}=0$, let
  \begin{align*}
    \phi_{\tup{a}}^{q,0} (\tup x):= \vartheta_{\tup{a}}^{q,0}(\tup
    x)\quad\wedge\quad \bigwedge_{\substack{\phi(\tup
        x)\in\Phi(\sigma,k,0)\\A\models\phi(\tup a)}}\phi(\tup{x}).
  \end{align*}
  Recall that~$\Phi(\sigma,k,0)$ denotes the (finite) set of all quantifier free
  normalised~$\FO[\sigma]$-formulas~$\phi(\tup x)$.
  For~$1\le{\ell}\leq q$, let
  \[
    \phi_{\tup{a}}^{q,{\ell}} (\tup x):= \vartheta_{\tup{a}}^{q,{\ell}}(\tup x)\wedge
    \bigwedge_{a\in
        V(A)}\exists
    x_{k+1}
    \phi_{\tup{a}a}^{q,{\ell}-1}(\tup{x},x_{k+1})\wedge 
    \forall x_{k+1}\bigvee_{a\in V(A)}\phi_{\tup{a}a}^{q,{\ell}-1}(\tup{x},x_{k+1}).
  \]
  If we remove repeated entries from the big conjunction and the big
  disjunction in the definition of~$\phi_{\tup{a}}^{q,{\ell}} (\tup
  x)$, we obtain a well-defined finite formula even for infinite
  structures~$A$. Moreover, it is easy to see that the~$q$-rank of
  this formula is~${\ell}$.  The following lemma implies
  Theorem~\ref{thm:game}.

\begin{lemma}\label{lem:game}
  Given~$q, 0\leq {\ell}\leq q, A, B$ and~$\tup{a}\in V(A)^{k}, \tup{b}\in
  V(B)^{k}$, the following are equivalent.
  \begin{enumerate}
  \item \label{item:ws}Duplicator has a winning strategy for
    the~${\ell}$-round EF$_q^+$ game on~$(A, \tup a,B,\tup b)$.
  \item \label{item:hint}$B\models\phi_{\tup{a}}^{q,{\ell}}(\tup{b})$.
  \item \label{item:equiv}$(A,\tup a)\equiv^+_{q, {\ell}} (B,\tup b)$.
  \end{enumerate}
\end{lemma}

\begin{proof}
  Assertion (\ref{item:equiv}) implies assertion (\ref{item:hint}), as
  the~$q$-rank of~$\phi_{\tup{a}}^{q,{\ell}}$ is~${\ell}$
  and~$A\models\phi_{\tup{a}}^{q,{\ell}}(\tup{a})$.

  Let~$q\in \N$. We prove the equivalence of (\ref{item:ws}) and
  (\ref{item:hint}) by induction on~${\ell}$.

  For~${\ell}=0$,~$(A,\tup a)\equiv^+_{q,{\ell}} (B,\tup b)$ if, and
  only if,~$\tup a\mapsto \tup b$ is a
  partial~$f_q(0)$-isomorphism. This is exactly the meaning
  of~$\phi_{\tup{a}}^{q,0}$.

  For~${\ell}>0$, 
  \begin{description}
  \item[] Duplicator has a winning strategy for the~${\ell}$-round
    EF$_q^+$ game on~$(A,\tup{a},B, \tup{b})$
    \item[$\Longleftrightarrow$] ~$\tup{a}$ and~$\tup{b}$ satisfy the same
    distance formulas up to~$f_q({\ell})$ and
    \begin{itemize}
    \item for every~$a\in V(A)$ there is a~$b\in V(B)$ such that
      Duplicator has a winning strategy for
      the~${\ell}-1$-round~$\FO_q^+$-game on~$(A, \tup{a}a)$
      and~$(B,\tup{b}b)$ and
    \item for every~$b\in V(B)$ there is an~$a\in V(A)$ such that
      Duplicator has a winning strategy for
      the~${\ell}-1$-round~$\FO_q^+$-game on~$(A, \tup{a}a)$
      and~$(B,\tup{b}b)$
    \end{itemize}
   \item[$\Longleftrightarrow$] ~$\tup{a}$ and~$\tup{b}$ satisfy the same
    distance formulas up to~$f_q({\ell})$ and
    \begin{itemize}
    \item for every~$a\in V(A)$ there is a~$b\in V(B)$ such
      that~$B\models\phi_{\tup{a}a}^{q,{\ell}-1}(\tup{b}b)$ and
    \item for every~$b\in V(B)$ there is an~$a\in V(A)$ such
      that~$B\models\phi_{\tup{a}a}^{q,{\ell}-1}(\tup{b}b)$ (by
      induction hypothesis)
    \end{itemize} 
  \item[$\Longleftrightarrow$]~$B\models\phi_{\tup{a}}^{q,{\ell}}(\tup{b})$
    (by construction of~$\phi_{\tup{a}}^{q,{\ell}}$).
  \end{description}
  It remains to show that (\ref{item:ws}) implies
  (\ref{item:equiv}). The proof is by induction
  on~${\ell}$. Case~${\ell}=0$ is handled as above. Let~${\ell}>0$ and
  suppose that Duplicator has a winning strategy for
  the~${\ell}$-round EF$_q^+$ game starting in position~$(A,\tup{a},B,
  \tup{b})$. Then the truth of atomic formulas and distances up
  to~$f_q({\ell})$ in~$\tup{a}$ and~$\tup{b}$ are preserved. Clearly,
  the set of formulas whose truth values are preserved is closed under
  negation and disjunction. Suppose that~$\phi(\tup{x})=\exists
  y\psi(\tup{x},y)$ and~$\phi$ is of rank at
  most~$(q,{\ell})$. Assume, for
  instance,~$A\models\phi(\tup{a})$. Then there is~$a\in V(A)$ such
  that~$A\models\phi(\tup{a}, a)$. By assumption Duplicator has a
  winning strategy for the~${\ell}$-round EF$_q^+$ game starting in
  position~$(A,\tup{a},B, \tup{b})$ and thus there is~$b\in V(B)$ such
  that Duplicator has a winning strategy for the~${\ell}-1$-round
  EF$_q^+$ game starting in position~$(A,\tup{a}a,B, \tup{b}b)$. Since
  the~$q$-rank of~$\psi$ is at most~${\ell}-1$, the induction
  hypothesis yields~$B\models\psi(\tup{b},b)$ and
  hence~$B\models\phi(\tup{b})$.
\end{proof}

\subsection{The Rank-Preserving Locality Theorem}
\label{sec:rpl}

We expand~$\sigma$-structures~$A$ by adding definable information
about neighbourhoods to every vertex. Let $\mathcal X$ be
an~$r$-neighbourhood cover of~$A$. For every~$a\in V(G)$, we fix some
cluster~${\mathcal X}(a)\in\XXX$ such that $N_r(a)\subseteq
{\XXX}(a)$. Actually, we view this assignments of clusters to the
vertices as being given with the neighbourhood cover. Formally, we
thus view an~$r$-neighbourhood cover~$\XXX$ as a mapping that
associates with every vertex~$a\in V(G)$ a set~${\XXX}(a)\subseteq
V(G)$ such that~$N_r(a)\subseteq {\XXX}(a)$.  For all~$q\in\N$,
let~$\sigma\star q$ be the vocabulary obtained from~$\sigma$ by adding
a fresh unary relation symbol $P_\phi$ for each
$\phi=\phi(x)\in\Phi^+(\sigma,1,q,q)$. For a~$\sigma$-structure $A$,
let~$A\star_{\XXX}q$ be the~$\sigma\star q$-expansion of~$A$ in
which~$P_\phi$ is interpreted by the set of all~$a\in V(A)$ such
that~$A\big[{\XXX}(a)\big]\models\phi(a)$.  We let~$\sigma\star^0
q:=\sigma$ and~$A\star_{\XXX}^0 q:=A$. For $i\ge 0$, we
let~$\sigma\star^{i+1} q:=(\sigma\star^i q)\star q$ and
$A\star_{\XXX}^{i+1} q:=\big(A\star_{\XXX}^{i}q\big)\star_{\XXX}q$.

A \emph{$(q,r)$-independence sentence} is a sentence of the form
\[
\exists x_1\ldots\exists x_q\Big(\bigwedge_{1\le i<j\le
  q}\dist(x_i,x_j)> 2r\wedge\bigwedge_{1\le i\le q}
\phi(x_i)\Big)
\]
for a quantifier-free first-order formula~$\phi(x_i)$. Note that the independence
sentences have the same form as the basic local sentences in Gaifman's
Theorem, except that the formula~$\phi(x)$ is required to be
quantifier-free, which implies that it is~$s$-local for every~$s\geq
0$. We denote the set of all~$(q,r)$-independence sentences of
vocabulary~$\sigma$ by~$\Psi(\sigma,q,r)$.

\begin{theorem}[Rank-Preserving Locality Theorem]\label{thm:rplocality}
  Let~$q\in\mathbb N$ and~$r=f_q(q)$. For
  every~$\FO[\sigma]$-formula~$\phi(x)$ of quantifier rank~$q$ there
  is an~$\FO^+[\sigma\star^{q+1} q]$-formula~$\hphi(x)$, which is a
  Boolean combination of~$(q+1,r)$-independence sentences and atomic
  formulas, such that for every~$\sigma$-structure~$A$,
  every~$r$-neighbourhood cover~$\XXX$ of~$A$, and every~$a\in V(A)$,
  \[
  A\models\phi(a)\iff A\star_{\XXX}^{q+1}q\models\hphi(a). 
  \]
  Furthermore,~$\hphi$ is computable from~$\phi$.
\end{theorem}

Even though we need the theorem in this general form, it may be
worthwhile to state, as a corollary, a version that does not refer to any
neighbourhood cover. It is obtained by applying the theorem to the generic
$r$-neighbourhood cover~$\XXX=\{N_r(v)\mid v\in V(G)\}$. We omit the
index~$\XXX$ in the~$\star$-notation when we refer to this
neighbourhood cover. As a further simplification, we only state the
corollary for sentences.

\begin{corollary}
  Let~$q\in\mathbb N$ and~$r=f_q(q)$. For
  every~$\FO[\sigma]$-sentence~$\phi$ of quantifier rank~$q$ there is
  an~$\FO^+[\sigma\star^{q+1} q]$-sentence~$\hphi$, which is a Boolean
  combination of~$(q+1,r)$-independence sentences, such that for
  every~$\sigma$-structure~$A$ and every~$a\in V(A)$,
  \[
  A\models\phi\iff A\star^{q+1}q\models\hphi. 
  \]
  Furthermore,~$\hphi$ is computable from~$\phi$.
\end{corollary}

To prove the theorem, it will be convenient to 
introduce the language
of types.  The \emph{$(q,{\ell})$-type} of a tuple~$\tup a\in V(A)^{k}$ in
a~$\sigma$-structure~$A$ is the set~$\tp_{q,{\ell}}(A,\tup
a)$
of all
formulas~$\phi(\tup x)\in\Phi^+(\sigma,k,q,{\ell})$ (normalised
$\FO^+[\sigma]$-formulas of~$q$-rank at most~${\ell}$) such that
$A\models\phi(\tup a)$. Note that
\[
(A,\tup a)\equiv^+_{q,{\ell}} (B,\tup b)\iff\tp_{q,{\ell}}(A,\tup
a)=\tp_{q,{\ell}}(B,\tup b).
\]
We call~$\atp_q(A,\tup a):=\tp_{q,0}(A,\tup
a)$
the \emph{atomic~$q$-type} of~$\tup a$ in~$A$.  We denote the set of
all~$(q,{\ell})$-types of~$k$-tuples in $\sigma$-structures
by~$T(\sigma,k,q,{\ell})$.

The \emph{$(q,r)$-independence type} of a structure~$A$ is the set
$\itp_{q,r}(A)$ of all~$(q',r')$-independence sentences for~$q'\le q$
and~$r'\leq r$
that are satisfied by~$A$. The set of all~$(q,r)$-independence types
of~$\sigma$-structures is denoted by~$I(\sigma,q,r)$.

\begin{lemma}\label{lem:rplocality}
  Let~$q\in\N$ and~$r:=f_q(q)$. Let~$A,B$ be~$\sigma$-structures
  and~$\XXX$,~$\YYY$~$r$-neighbourhood covers of~$A,B$,
  respectively. Let~$a_0\in V(A), b_0\in V(B)$ such that
  \begin{align*}
    & \itp_{q+1,r}\big(A\star^q_{\XXX} q\big)
    =\itp_{q+1,r}\big(B\star^q_{\YYY}q\big)\\
    \text{and }& \atp_q(A\star^{q+1}_{\XXX} q , a_0)=\atp_q(B\star^{q+1}_{\YYY}q ,b_0).
  \end{align*}
  Then~$(A,a_0)\equiv^+_{q,q} (B,b_0)$.
\end{lemma}

\begin{proof}
  We start by fixing some notation.  For~$0\le k\le q$, we
  let~$\sigma_k:=\sigma\star^{q-k} q$ and~$A_k:=A\star^{q-k}_{\XXX} q$
  and~$B_k:=B\star^{q-k}_{\YYY}q$ and~$r_k:=f_q(q-k)$.  Throughout the
  proof,~$\tup x$ always denotes a tuple~$(x_0,\ldots,x_k)$ (for
  varying~$k$), and similarly~$\tup a,\tup b$ denote
  tuples~$(a_0,\ldots,a_k)$ and~$(b_0,\ldots,b_k)$.  We
  write~$J\sqsubseteq H$ to denote that~$J$ is a connected component
  of a graph~$H$. Furthermore, if~$V(H)=\{0,\ldots,k\}$
  and~$J\sqsubseteq H$, then~$\tup x_J$ denotes the sub-tuple of~$\tup
  x$ with entries~$x_j$ for~$j\in V(J)$, and~$\tup a_J,\tup b_J$
  denote the corresponding sub-tuples of~$\tup a,\tup b$.

  We shall prove that Duplicator has a winning strategy for
  the~$q$-round EF$_q^+$ game on~$(A,a_0,B,b_0)$.  We describe a
  winning strategy for Duplicator satisfying the following conditions
  for every position~$p=(A,\tup a,B,\tup b)$, where~$\tup a=(a_0,
  a_1,\ldots,a_k)$ and~$\tup b=(b_0, b_1,\ldots,b_k)$, of the game
  that can be reached if Duplicator plays according to this
  strategy. Let~$H_p$ be the graph with vertex set~$V(H_p)=\{0,\ldots,
  k\}$ and edge set
  \[
  E(H_p):=\big\{ij\st \dist(a_i,a_j)\le r_k\text{ or
  }\dist(b_i,b_j)\le r_k\big\}.
  \]
  Then for every component~$J\sqsubseteq H_p$ there are induced
  substructures~$A_J\subseteq A_k$,~$B_J\subseteq B_k$ such that the
  following conditions are satisfied:
  \begin{enumerate}
  \item [(i)]~$N_{r_k}(a_j)\subseteq V(A_J)$ and~$N_{r_k}(b_j)\subseteq
    V(B_J)$ for all~$j\in V(J)$;
  \item[(ii)]~$\big(A_J,\tup a_J\big)\equiv_{q,q-k}^+
    \big(B_J,\tup b_J\big)$.
  \end{enumerate}
  Note that this implies that~$\tup{a}\mapsto \tup{b}$ is
  a partial~$f_q(q-k)$-isomorphism.

  The proof is by induction on~$k$. For the base step~$k=0$, note that
  the graph~$H:=H_{p}$ is the one-vertex graph, which is connected. We
  let~$A_H:=A_0\big[\XXX(a_0)\big]$ and~$B_H:=B_0\big[\YYY(b_0)\big]$.
  Then (i) holds, because~$\XXX,\YYY$ are~$r$-neighbourhood covers
  and~$r=r_0$. By the assumption of the lemma, we
  have~$\atp_q(A_0\star_\XXX^{q+1} q,a_0)=\atp_q(B_0\star_\YYY^{q+1}
  q,b_0)$. In particular, for every
  formula~$\phi(x)\in\Phi^+(\sigma_0,1,q,q)$ we
  have~$A_0\star_\XXX^{q+1} q\models P_\phi(a_0)\iff
  B_0\star_\YYY^{q+1} q\models P_\phi(b_0)$, which
  implies~$A_H\models\phi(a_0)\iff B_H\models\phi(b_0)$ by the
  definition of the~$\star$-operator. As
  every~$\FO^+[\sigma_0]$-formula~$\phi(x_0)$ of~$q$-rank at most~$q$
  is equivalent to a formula in~$\Phi^+(\sigma_0,1,q,q)$, this
  implies~$(A_H,a_0)\equiv^+_{q,q}(B_H,b_0)$, that is, assertion (ii).

  For the inductive step, suppose that we are in a position~$p=(A,\tup
  a,B,\tup b)$, where~$\tup a=(a_0,a_1,\ldots,a_k)$ and~$\tup
  b=(b_0,b_1,\ldots,b_k)$ for some~$k<q$. Again, let~$H :=
  H_p$. Suppose that in the~$(k+1)$st round of the game, Spoiler
  picks~$a_{k+1}\in V(A)$.
  \begin{cs}
    \case1~$\dist(a_{k+1},a_i)\le r_{k}$ for some~$i\in \{0,\ldots,k\}$.\\
    Let~$I\sqsubseteq H$ be the connected component of~$i$, and
    let~$A_I\subseteq A_k$,~$B_I\subseteq B_k$ be substructures
    satisfying (i) and (ii).  By (i),~$a_{k+1}\in V(A_I)$. By (ii),~$(A_I,
    \tup{a}_I)\equiv_{q,q-k}^+(B_I, \tup{b}_I)$, and thus Duplicator
    has a winning strategy for
    the~$q-k$-round EF$_q^+$-game on~$(A_I, \tup{a}_I,
    B_I,\tup{b}_I)$.  Let~$b_{k+1}$ be Duplicator's answer if Spoiler
    picks~$a_{k+1}$ in this game. Then
    \begin{equation}
      \label{eq:a}
      (A_I,\bar a_Ia_{k+1})\equiv^+_{q,q-k-1} (B_I,\bar b_Ib_{k+1}).
    \end{equation}
    This implies~$\atp_q(A_k,a_{k+1})=\atp_q(B_k,b_{k+1})$ and thus
    \begin{equation}
      \label{eq:b}
      (A_{k+1}[\XXX(a_{k+1})],a_{k+1})\equiv^+_{q,q-k-1}(B_{k+1}[\YYY(b_{k+1})],b_{k+1}).
    \end{equation}
    We choose~$b_{k+1}$ as Duplicator's answer in the game
    on~$A,B$. Thus the new position is
    \[
    p':=(A,\tup aa_{k+1},B,\tup
    bb_{k+1}).
    \]
    Let~$H':=H_{p'}$. 
    \begin{cs}
      \case{1a}~$\dist(a_{k+1},a_i)\le r_{k+1}$ for some~$i\in
      \{0,\ldots,k\}$.

      Then 
      \begin{equation}
        N_{r_{k+1}}(a_{k+1})\subseteq N_{r_k}(\tup a)\subseteq V(A_I),\label{eq:c}
      \end{equation}
      because~$r_k\ge 2r_{k+1}$, and
      \begin{equation}
        \label{eq:d}
        N_{r_{k+1}}(b_{k+1})\subseteq
        N_{r_k}(\tup b)\subseteq V(B_I),
      \end{equation}
      because~$(q,q-k-1)^+$-equivalence preserves distances up
      to~$r_{k+1}$.

      Let~$J'\sqsubseteq H'$. Then there is a~$J\sqsubseteq H$ such
      that~$V(J')\cap\{0,\ldots,k\}\subseteq V(J)$. To see this, just
      note that if~$j(k+1)\in E(H')$ and~$(k+1)j'\in E(H')$
      then~$jj'\in E(H)$, because~$2r_{k+1}\le r_k$. Thus, whenever
      there is a path between two vertices~$j,j'\in\{0,\ldots,k\}$
      in~$H'$ there also is a path in~$H$. We let~$A_{J'}\subseteq
      A_{k+1}$ be the restriction of~$A_J\subseteq A_k$
      to~$\sigma_{k+1}$ and~$B_{J'}\subseteq B_{k+1}$ the restriction
      of~$B_J\subseteq A_k$ to~$\sigma_{k+1}$. Then if~$J=I$ and
      hence~$k+1\in V(J')$, (i) for~$p'$ and~$A_{J'},B_{J'}$ follows
      from \eqref{eq:c} and \eqref{eq:d}, and (ii) follows from
      \eqref{eq:a}. If~$J\neq I$, then (i) and (ii) for~$p'$
      and~$A_{J'},B_{J'}$ are inherited from (i) and (ii) for~$p$
      and~$A_J,B_J$.

      \case{1b}~$\dist(a_{k+1},a_i)> r_{k+1}$ for all~$i\in \{0,\ldots, k\}$.\\
      Let~$J'\sqsubseteq H'$. Then either~$V(J')=\{k+1\}$, or there is
      a~$J\sqsubseteq H$ such that~$V(J')\subseteq
      V(J)$. If~$V(J')=\{k+1\}$, we
      let~$A_{J'}:=A_{k+1}\big[\XXX(a_{k+1})\big]$
      and~$B_{J'}:=B_{k+1}\big[\YYY(b_{k+1})\big]$. Then (i) holds
      because~$\XXX$ and~$\YYY$ are~$r$-neighbourhood covers, and (ii)
      follows from \eqref{eq:b}. If there is a connected component~$J$
      of~$H$ such that~$V(J')\subseteq V(J)$, we let~$A_{J'}\subseteq
      A_{k+1}$ be the restriction of~$A_J\subseteq A_k$
      to~$\sigma_{k+1}$ and~$B_{J'}\subseteq B_{k+1}$ the restriction
      of~$B_J\subseteq B_k$ to~$\sigma_{k+1}$. Then (i) and (ii)
      for~$p'$ and~$A_{J'},B_{J'}$ are inherited from (i) and (ii)
      for~$p$ and~$A_J,B_J$.
    \end{cs}

    \case2~$\dist(a_{k+1},a_i)>r_k$ for all~$i\in \{0,\ldots, k\}$.\\
    Let~$t:=\atp_q(A_k,a_{k+1})$. We will prove the existence of
    a~$b_{k+1}\in V(B)$ with~$\atp_q(B_k,b_{k+1})=t$
    and~$\dist(b_{k+1},b_i)> r_{k+1}$ for all~$i\in \{0,\ldots,
    k\}$. We can then argue as in Case~1b. Assume towards a
    contradiction that
    \begin{enumerate}
    \item[(A)] there is no~$b\in V(B)$ with~$\atp_q(B_k,b)=t$ and
      ~$\dist(b,b_i)> r_{k+1}$ for all~$i\in \{0,\ldots, k\}$.
    \end{enumerate}

    The first step is to construct~$d,D,\ell$ such that~$2r_{k+1}\le
    d\le D-4r_{k+1}$ and~$D\le r_k$ and~$\ell\le k$ and there are
    elements~$a^0,\ldots,a^\ell\in V(A)$ with~$\atp_q(A_k,a^i)=t$
    and~$\dist(a^i,a^j)>D$ for~$i\neq j\in\{0,\ldots,\ell\}$, but no
    elements~$a_*^0,\ldots,a_*^{\ell+1}\in V(A)$
    with~$\atp_q(A_k,a_*^i)=t$ and~$\dist(a_*^i,a_*^j)>d$ for~$i\neq
    j\in\{0,\ldots, \ell+1\}$.

    We let~$d_0:=2r_{k+1}$, and we let~$\ell_0$ be maximal such that
    there are~$a^0_0,\ldots,a^{\ell_0}_0$ with~$\atp_q(A_k,a_0^i)=t$
    for all~$i\in\{0,\ldots, \ell_0\}$ and~$\dist(a_0^i,a_0^j)>d_0$
    for all~$i\neq j\in\{0,\ldots, \ell_0\}$. Suppose first
    that~$\ell_0>k$. As~$A$ and~$B$ satisfy the
    same~$(k+1,d_0/2)$-independence sentences (note that~$d_0$ is
    even), there are elements~$b_0^0,\ldots,b_0^{k+1}\in V(B)$
    with~$\atp_q(B_k,b_0^i)=t$ for all~$i\in\{0,\ldots, k+1\}$
    and~$\dist(b_0^i,b_0^j)>d_0$. By (A), for every~$i\in\{0,\ldots,
    k+1\}$ there is a~$j(i)\in\{0,\ldots, k\}$ such
    that~$\dist(b^i_0,b_{j(i)})\le
    r_{k+1}=d_0/2$. As~$\dist(b_0^i,b_0^j)>d_0$, we have~$j(i)\neq
    j(i')$ for~$i\neq i'\in\{0,\ldots, k+1\}$. This is a
    contradiction, which proves that~$\ell_0\le k$.

    Now suppose that~$d_h,\ell_h$ are defined for some~$h\ge
    0$. Let~$d_{h+1}:=d_h+4r_{k+1}$, and let~$\ell_{h+1}$ be maximal
    such that there are~$a^0_{h+1},\ldots,a^{\ell_{h+1}}_{h+1}$
    with~$\atp_q(A_k,a_{h+1}^i)=t$ for all~$i\in\{0,\ldots,
    \ell_{h+1}\}$ and~$\dist(a_{h+1}^i,a_{h+1}^j)>d_{h+1}$ for
    all~$i\neq j\in\{0,\ldots,
    \ell_{h+1}\}$. Then~$\ell_{h+1}\le\ell_h$. If~$\ell_{h+1}=\ell_h$
    for the first time, we stop the construction. Then~$h\le k$ and
    thus~$d_{h+1}=(4(h+1)-2)r_{k+1}\le r_k$.  We let~$d:=d_h$
    and~$D:=d_{h+1}$ and~$\ell:=\ell_h=\ell_{h+1}$.
    
    As~$A$ and~$B$ satisfy the same~$(k+1,D/2)$-independence
    sentences, there are elements~$b^0,\ldots,b^\ell\in V(B)$
    with~$\atp_q(B_k,b^i)=t$ and~$\dist(b^i,b^j)>D$.  Then for
    every~$i\in\{0,\ldots, \ell\}$ there is a~$j(i)\in\{0,\ldots,k\}$
    such that~$\dist(b^i,b_{j(i)})\le r_{k+1}$. The~$j(i)$ are
    mutually distinct, because~$\dist(b^i,b^j)>2r_{k+1}$ for~$i\neq
    j$.  To simplify the notation, let us assume that~$j(i)=i$ for
    all~$i\in\{0,\ldots, \ell\}$. As~$\dist(b^i,b^j)>D$, we
    have~$\dist(b_i,b_j)>D-2r_{k+1}$. Then it follows from (ii)
    that~$\dist(a_i,a_j)>D-2r_{k+1}$, because~$D-2r_{k+1}\le r_k$. It
    also follows from (ii) that for all~$i\in\{0,\ldots, \ell\}$ there
    is an~$a^i_*$ such that~$\dist(a^i_*,a_i)\le r_{k+1}$
    and~$\atp_q(A_k,a^i_*)=t$. Then for~$i\neq j$ we
    have~$\dist(a^i_*,a^j_*)>D-4r_{k+1}\ge d$. Furthermore, we
    have~$\dist(a_{k+1},a^i_*)>r_k-r_{k+1}\ge
    d$. Letting~$a_*^{\ell+1}:=a_{k+1}$, we have
    found~$a_*^1,\ldots,a_*^{\ell+1}\in V(A)$
    with~$\atp_q(A_k,a_*^i)=t$ and~$\dist(a_*^i,a_*^j)>d$. This is a
    contradiction.
\end{cs}
\end{proof}

We will show next how the Rank Preserving Locality Theorem follows from this lemma by
standard techniques from logic.
\medskip

\begin{proof}[Proof of the Rank Preserving Locality Theorem]
  Let~$\phi(x)\in\FO[\sigma]$ be a first-order formula of quantifier
  rank~$q$. Let~$r:=f_q(q)$ and~$\sigma_I:=\sigma\star^qq$
  and~$\sigma_T:=\sigma\star^{q+1} q$. Furthermore,
  let~$I:=I(\sigma_I,q+1,r)$ and~$T:=T(\sigma_T,1,q,0)$. A
  pair~$(\eta,\theta)\in I\times T$ is \emph{satisfiable} if there are
  a~$\sigma$-structure~$A$ and an~$r$-neighbourhood cover~$\XXX$
  of~$A$ and an~$a\in V(A)$ such
  that~$\itp_{q+1,r}(A\star^q_{\XXX}q)=\eta$ and~$\atp_q
  (A\star^{q+1}_{\XXX}q,a)=\theta$.

  It follows from Lemma~\ref{lem:rplocality} that for all satisfiable
  pairs~$(\eta,\theta)\in I\times T$ the following two statements are
  equivalent.
  \begin{enumerate}
  \item[(A)] There are a~$\sigma$-structure~$A$ and
    an~$r$-neighbourhood cover~$\XXX$ of~$A$ and an~$a\in V(A)$ such
    that~$\itp_{q+1,r}(A\star^q_{\XXX}q)=\eta$ and~$\atp_q
    (A\star^{q+1}_{\XXX}q,a)=\theta$ and ~$A\models\phi(a)$.
  \item[(B)] For all~$\sigma$-structures~$A$ and~$r$-neighbourhood
    covers~$\XXX$ of~$A$ and~$a\in V(A)$,
    if~$\itp_{q+1,r}$ $(A\star^q_{\XXX}q)=\eta$ and~$\atp_q
    (A\star^{q+1}_{\XXX}q,a)=\theta$, then~$A\models\phi(a)$.
  \end{enumerate}
  Thus there is a subset~$S_\phi\subseteq I\times T$ such that for
  all~$\sigma$-structures~$A$, all~$r$-neighbourhood covers~$\XXX$
  of~$A$, and all~$a\in V(A)$,
  \begin{equation}\label{eq:rpl1}
  A\models\phi(a)\iff\exists (\eta,\theta)\in S_\phi:\;\itp_{q+1,r}(A\star^q_{\XXX}q)=\eta\text{ and }\atp_q (A\star^q_{\XXX}q,a)=\theta.
  \end{equation}

  Recall that every~$(q+1,r)$-independence type~$\eta\in
  I$ is a subset of the finite set $\Psi(\sigma_I,q+1,r)$,
  and for every~$\sigma_I$-structure~$A$ we have
  \[
  \itp_{q+1,r}(A)=\eta\iff A\models\bigwedge_{\psi\in\eta}\psi\wedge\bigwedge_{\psi\in
    \Psi(\sigma_I,q+1,r)\setminus\eta}\neg\psi.
  \]
  We denote the
  sentence~$\bigwedge_{\psi\in\eta}\psi\wedge\bigwedge_{\psi\in
    \Psi(\sigma_I,q+1,r)\setminus\eta}\neg\psi$ by~$\widetilde\eta$
  and say that it \emph{defines} the type~$\eta$. But we can actually
  define~$\widetilde\eta$ for every subset~$\eta\subseteq
  \Psi(\sigma_I,q+1,r)$. Then either~$\widetilde\eta$ is unsatisfiable
  or there is some~$\sigma_I$-structure~$A$ such
  that~$\itp_{q+1,r}(A)=\eta$.

  Similarly, every atomic type~$\theta\in T(\sigma_T,1,q,0)$ is a
  subset of the finite set~$\Phi^+(\sigma_T,1,q,0)$, and for
  every~$\sigma_T$-structure~$A$ and every~$a\in V(A)$ we have
  \[
  \atp_q(A,a)=\theta\iff A\models \bigwedge_{\zeta(x)\in\theta}\zeta(a)\wedge\bigwedge_{\zeta(x)\in
    \Phi(\sigma_T,1,q,0)\setminus\theta}\neg\zeta(a).
  \]
  We denote the
  formula~$\bigwedge_{\zeta(x)\in\theta}\zeta(x)\wedge\bigwedge_{\zeta(x)\in
    \Phi(\sigma_T,1,q,0)\setminus\theta}\neg\zeta(x)$
  by~$\widetilde\theta(x)$. Again, we can define~$\widetilde\theta(x)$
  for every subset~$\theta\subseteq\Phi^+(\sigma_T,1,q,0)$. Then
  either~$\widetilde\theta(x)$ is unsatisfiable, or there is
  some~$\sigma_T$-structure~$A$ and~$a\in V(A)$ such
  that~$\atp_{q}(A,a)=\theta$.

  It follows from \eqref{eq:rpl1} that for
  all~$\sigma$-structures~$A$, all~$r$-neighbourhood covers~$\XXX$
  of~$A$, and all~$a\in V(A)$,
    \begin{equation}\label{eq:rpl2}
  A\models\phi(a)\iff A\star^{q+1}_{\XXX}q\models\bigvee_{(\eta,\theta)\in S_\phi}
  \big(\widetilde\eta\wedge\widetilde\theta(a)\big).
  \end{equation}
  Here we use that the~$\sigma_T$-structure~$A\star^{q+1}_{\XXX}q$ is
  an expansion of the~$\sigma_I$-structure~$A\star^{q}_{\XXX}q$.

  We could let~$\hphi(x)=\bigvee_{(\eta,\theta)\in S_\phi}
  \big(\widetilde\eta\wedge\widetilde\theta(x)\big)$. Clearly, this
  formula has the desired syntactic form, and by \eqref{eq:rpl2}
  satisfies the assertion of the theorem. However, we want~$\hphi(x)$
  to be computable from~$\phi(x)$, and with this definition, it is
  not, because the choice of~$S_\phi$ is not unique and, so far,
  arbitrary. However, we will prove that we can compute some
  set~$S_\phi$ satisfying \eqref{eq:rpl2}.

  \medskip We need to incorporate the~$r$-neighbourhood covers into
  the logical framework. Let~$R$ be a fresh binary relation symbol
  and~$\sigma_R:=\sigma\cup\{R\}$. For every~$\sigma$-structure~$A$
  and every mapping~$\XXX:V(A)\to 2^{V(A)}$, we let~$A^{\XXX}$ be
  the~$\sigma\cup\{R\}$-expansion of~$A$ with
  \[
  R(A^{\XXX})=\{ab\mid b\in\XXX(a)\}.
  \]
  Recall that we view~$r$-neighbourhood covers of~$A$ as
  mappings~$\XXX:V(A)\to 2^{V(A)}$ where~$N_r(a)\subseteq\XXX(a)$ for
  each~$a\in V(A)$. We let~$\gamma:=\forall x\forall y(\dist(x,y)\le
  r\longrightarrow R(x,y))$. Then~$\XXX$ is an~$r$-neighbourhood cover
  of~$A$ if, and only if,~$A^\XXX\models\gamma$. It is not hard to see
  that the structure~$A\star_{\XXX}q$ is definable within~$A^{\XXX}$,
  which means that for every (unary) relation symbol~$P\in(\sigma\star
  q)\setminus\sigma$ there is a~$\sigma\cup\{R\}$-formula~$\chi_P(x)$
  such that~$P(A\star_{\XXX} q)=\{a\in V(A)\mid
  A^{\XXX}\models\chi_P(a)\}$. By the so-called Lemma on Syntactical
  Interpretations (see~\cite{EbbinghausFluTho94}), this implies that
  for every~$\sigma\star q$-formula~$\psi(x)$ there is
  a~$\sigma\cup\{R\}$-formula~$\psi_R(x)$ such that~$A\star_{\XXX}
  q\models\psi(a)\iff A^{\XXX}\models\psi_R(a)$. Using this, we can
  inductively prove that~$A\star^\ell_{\XXX}q$ is definable
  within~$A^{\XXX}$ and that for every~$\sigma\star^\ell
  q$-formula~$\psi(x)$ there is
  a~$\sigma\cup\{R\}$-formula~$\psi_R(x)$ such
  that~$A\star^\ell_{\XXX} q\models\psi(a)\iff
  A^{\XXX}\models\psi_R(a)$. In particular, for
  every~$\eta\subseteq\Psi(\sigma_I,q+1,r)$ there is
  a~$\sigma_R$-sentence~$\widetilde\eta_R$ such
  that~$A\star_{\XXX}^{q+1}q\models\widetilde\eta\iff
  A^{\XXX}\models\widetilde\eta_R$ and for
  every~$\theta(x)\subseteq\Phi(\sigma_T,1,q,0)$ there is
  a~$\sigma_R$-sentence~$\widetilde\theta_R(x)$ such
  that~$A\star_{\XXX}^{q+1}q\models\widetilde\theta(a)\iff
  A^{\XXX}\models\widetilde\theta_R(a)$.

  It follows from \eqref{eq:rpl2} that for
  all~$\sigma$-structures~$A$, all~$r$-neighbourhood covers~$\XXX$
  of~$A$, and all~$a\in V(A)$,
    \begin{equation}\label{eq:rpl3}
  A\models\phi(a)\iff A^{\XXX}\models\bigvee_{(\eta,\theta)\in S_\phi}
  \big(\widetilde\eta_R\wedge\widetilde\theta_R(a)\big).
  \end{equation}
  As~$A^{\XXX}$ is an expansion of~$A$, on the left-hand side of
  \eqref{eq:rpl3} we can replace~$A$ by~$A^{\XXX}$ and thus
  rewrite \eqref{eq:rpl3} as
    \begin{equation}\label{eq:rpl4}
  A^{\XXX}\models\phi(a)\longleftrightarrow\bigvee_{(\eta,\theta)\in S_\phi}
  \big(\widetilde\eta_R\wedge\widetilde\theta_R(a)\big).
  \end{equation}
  Recalling that a~$\sigma_R$-structure~$A_R$ equals~$A^{\XXX}$ for
  some~$r$-neighbourhood cover~$\XXX$ of a $\sigma$-structure~$A$ if
  any only if~$A_R\models\gamma$, for all~$\sigma_R$-structures~$A_R$
  and all~$a\in V(A_R)$ we thus have
  \begin{equation}\label{eq:rpl5}
  A_R\models\gamma\longrightarrow\Big( \phi(a)\longleftrightarrow\bigvee_{(\eta,\theta)\in S_\phi}
  \big(\widetilde\eta_R\wedge\widetilde\theta_R(a)\big)\Big).
  \end{equation}
  For every subset~$S\subseteq I\times T$, let 
  \[
  \alpha_S(x)=\gamma\longrightarrow\Big( \phi(x)\longleftrightarrow\bigvee_{(\eta,\theta)\in S}
  \big(\widetilde\eta_R\wedge\widetilde\theta_R(x)\big)\Big).
  \]
  By \eqref{eq:rpl5}, the formula~$\alpha_{S_{\phi}}(x)$ is
  valid. Note that so far we thought of~$\alpha_S(x)$ as
  an~$\FO^+$-formula, but we can directly translate
  every~$\FO^+$-formula into an equivalent~$\FO$-formula by
  substituting appropriate distance formulas for the distance
  atoms. This changes the rank, but at this point we no longer care
  about the rank. Thus we view~$\alpha_S(x)$ as
  an~$\FO[\sigma_R]$-formula.

  The set of all valid~$\FO[\sigma_R]$-formulas is recursively
  enumerable. We start an enumeration algorithm and wait for the first
  formula~$\alpha_S(x)$ it produces. This will happen eventually,
  because we know that~$\alpha_{S_\phi}(x)$ is valid. The
  set~$S\subseteq I\times T$ of the first formula~$\alpha_S(x)$
  returned by enumeration algorithm is not necessarily the same as the
  set~$S_\phi$ we started with. However, by retracing our construction
  backwards, it is easy to see that~$S$ satisfies \eqref{eq:rpl2},
  that is, for all~$\sigma$-structures~$A$, all~$r$-neighbourhood
  covers~$\XXX$ of~$A$, and all~$a\in V(A)$,
    \begin{equation*}
  A\models\phi(a)\iff A\star^{q+1}_{\XXX}q\models\bigvee_{(\eta,\theta)\in S}
  \big(\widetilde\eta\wedge\widetilde\theta(a)\big).
  \end{equation*}
  We  define~$\hphi(x) :=\bigvee_{(\eta,\theta)\in S}
  \big(\widetilde\eta\wedge\widetilde\theta(x)\big)$. As argued above, this
  formula satisfies the conditions of the theorem, and by construction
  it is computable from~$\phi(x)$.
  
  Note that if, given a formula~$\phi$, we first compute an equivalent
  normalised formula~$\phi'$ and then apply the procedure above
  to~$\phi'$, then we can compute an upper bound for the running time.
\end{proof}

\section{The Main Algorithm}
\label{sec:mc}

We are now ready to prove our main result,
Theorem~\ref{theo:main}. We actually prove a slightly more general
theorem. A
\emph{coloured-graph vocabulary} consists of the binary relation
symbol~$E$ and possibly finitely many unary relation symbols. In
particular, if~$\sigma$ is a coloured-graph vocabulary then
$\sigma\star q$ (as defined in Section~\ref{sec:rpl}) is a coloured
graph vocabulary. A \emph{$\sigma$-coloured graph} is a
$\sigma$-structure whose~$\{E\}$-restriction is a simple undirected
graph.\footnote{To see that this is consistent with the definition of
  coloured graphs in Section~\ref{sec:scattered-sets}, we may define
  the \emph{colour} of a vertex $v$ in a $\sigma$-coloured graph $G$
  to be the set of all unary relation symbols $P\in\sigma$ such that
  $v\in P(G)$.}
We call the~$\{E\}$-restriction of a~$\sigma$-colored graph the
\emph{underlying graph} of~$G$.

\begin{theorem}\label{theo:main+u}
  For every nowhere dense class~$\mathcal C$, every~$\epsilon>0$,
  every coloured graph vocabulary~$\sigma$, and every first-order
  formula~$\phi(x)\in\FO[\sigma]$, there is an algorithm that, given
  a~$\sigma$-coloured graph~$G$ whose underlying graph is in~$\CCC$,
  computes the set of all~$v\in V(G)$ such that~$G\models\phi(v)$ in
  time~$\Oof(n^{1+\epsilon})$.

  Furthermore, if~$\CCC$ is effectively nowhere dense, then there is a
  computable function~$f$ and an algorithm that, given~$\epsilon>0$, a
  formula~$\phi(x)\in\FO[\sigma]$ for some coloured-graph
  vocabulary~$\sigma$, and a~$\sigma$-coloured graph~$G$, computes the
  set of all~$v\in V(G)$ such that~$G\models\phi(v)$ in
  time~$f(|\phi|, \epsilon)\cdot n^{1+\epsilon}$.
\end{theorem}

Clearly, this implies Theorem~\ref{theo:main}.

We need one more lemma for the proof. It describes a standard
reduction that allows us to remove a bounded number of elements from a
structure in which we want to evaluate a formula.

\begin{lemma}\label{lem:redlem}
  Let~$\sigma$ be a coloured-graph vocabulary
  and~$k,\ell,m,q\in\mathbb N$ with~$0\le\ell\le q$. Then there are
  \begin{enumerate}
  \item a coloured-graph vocabulary~$\sigma'\supseteq\sigma$,
  \item for
    every~$\FO^+[\sigma]$-formula~$\phi(x_1,\ldots,x_k,y_1,\ldots,y_m)$
    of~$q$-rank~$\ell$ and every atomic~$q$-type~$\theta\in
    T(\sigma,m,q,0)$ an~$\FO^+[\sigma']$
    formula~$\phi^\theta(x_1,\ldots,x_k)$ of~$q$-rank at most~$\ell$,

  \item for every~$\sigma$-coloured graph~$G$ and
    all~$w_1,\ldots,w_m\in V(G)$ a~$\sigma'$-expansion~$G'$
    of~$G\setminus\{w_1,\ldots,w_m\}$,
  \end{enumerate}
  such that if~$\atp_q(G,w_1,\ldots,w_m)=\theta$ then for
  all~$v_1,\ldots,v_k\in V(G)\setminus\{w_1,\ldots,w_m\}$
  \[
  G\models\phi(v_1,\ldots,v_k,w_1,\ldots,w_m)\iff
  G'\models\phi^\theta(v_1,\ldots,v_k).
  \]
  Furthermore,~$\phi^\theta$ is computable from~$\phi$ and~$\theta$,
  and~$G'$ is computable from~$G$ and~$w_1,\ldots,w_m$ in
  time~$f(\ell,m,q)\cdot(|V(G)|+|E(G)|)$.
\end{lemma}

\begin{proof}
  We use a game theoretic argument similar to (but simpler than) the
  proof of the rank preserving locality theorem.

  For~$1\le i\le f_q(\ell)$ and~$1\le j\le m$, we let~$Q_{ij}$ be a
  fresh unary relation symbol, and we let~$\sigma'$ be the union
  of~$\sigma$ with all these~$Q_{ij}$. For every~$\sigma$-coloured
  graph~$G$ and all~$w_1,\ldots,w_m\in V(G)$ we let~$G'$ be
  the~$\sigma'$-expansion of~$G\setminus\{w_1,\ldots,w_m\}$ with
  \[
  Q_{ij}(G')=\{v\in V(G)\setminus\{w_1,\ldots,w_m\}\mid\dist^G(v,w_j)=i\}.
  \]
  Clearly,~$G'$ can be computed from~$G$ in
  time~$f(\ell,m,q)\cdot(|V(G)|+|E(G)|)$, for some function~$f$.

  \begin{Claim}\sloppy
    Let~$G_1,G_2$ be~$\sigma$-coloured graphs and~$v_{11},\ldots,v_{1k},w_{11},\ldots,w_{1m}\in
    V(G_1)$,~$v_{21},\ldots,v_{2k},w_{21},\ldots,w_{2m}\in V(G_2)$ such that
    \[
    \atp_q(G_1,w_{11},\ldots,w_{1m})=\atp_q(G_2,w_{21},\ldots,w_{2m})\]
    and
    \[
    G_1',(v_{11},\ldots,v_{1k})\equiv^+_{(q,\ell)}G_2',(v_{21},\ldots,v_{2k}).
    \]
    Then
    \[
    G_1,(v_{11},\ldots,v_{1k},w_{11},\ldots,w_{1m})\equiv^+_{q,\ell}G_2,(v_{21},\ldots,v_{2k},w_{21},\ldots,w_{2m}).
    \]
  \end{Claim}
  
  \begin{ClaimProof}
    It is easy to see that Duplicator has a winning strategy for
    the~$\ell$-round~$EF^+_q$-game
    on $(G,(v_{11},\ldots,v_{1k},$ $w_{11},\ldots,w_{1m}),
    G_2,(v_{21},\ldots,v_{2k},w_{21},\ldots,w_{2m}))$: she simply
    plays according to a winning strategy for
    the $\ell$-round~$EF^+_q$-game
    on~$(G_1',(v_{11},\ldots,v_{1k}),G_2',(v_{21},\ldots,v_{2k}))$,
    and whenever Spoiler selects a~$w_{ij}$ she answers by
    selecting~$w_{(3-i)j}$.
  \end{ClaimProof}
  
  The claim implies that there is a set~$S_{\phi,\theta}\subseteq
  T(\sigma',k,q,\ell)$ such that
  \[
  G\models\phi(v_1,\ldots,v_k,w_1,\ldots,w_m)\iff G'\models\bigvee_{\eta\in
    S_{\phi,\theta}}\bigwedge_{\psi(x_1,\ldots,x_k)\in\eta}\psi(v_1,\ldots,v_k).
  \]
  It remains to prove that we can compute such a set~$S_{\phi,\theta}$
  from~$\phi$ and~$\theta$. We use an argument based on the recursive
  enumerability of the valid first-order sentences similar to the one
  in the proof of the Rank Preserving Locality Theorem.
\end{proof}

\bigskip
\begin{proof}[Proof of Theorem~\ref{theo:main+u}]
  Let~$\CCC$ be a nowhere dense class of graphs and
  $\epsilon>0$. Without loss of generality we may assume
  that~$\epsilon\le 1/2$, which implies $\epsilon^2\le\epsilon/2$, and
  that~$\CCC$ is closed under taking subgraphs.

  The input to our algorithm is an~$\epsilon\le 1/2$,
  a~$\sigma$-coloured graph~$G$ whose $\{E\}$-restriction is in~$\CCC$
  and an~$\FO^+[\sigma]$-formula~$\phi(x)$, for some coloured-graph
  vocabulary~$\sigma$. Our algorithm will compute the set of all~$v\in
  V(G)$ such that~$G\models\phi(v)$ in time $\Oof(n^{1+\epsilon})$.

  We start by fixing a few parameters. We choose~$q$ such that the
  $q$-rank of~$\phi$ is at most~$q$ and let~$r=f_q(q)$. By the
  Rank-Preserving Locality Theorem, we can find an
  $\FO^+[\sigma\star^{q+1}q]$-formula~$\hphi(x)$, which is a Boolean
  combination of~$(q+1,r)$-independence sentences and atomic formulas,
  such that for all~$\sigma$-coloured graphs~$G$,
  all~$r$-neighbourhood covers~$\XXX$ of~$G$, and all~$v\in V(G)$ we
  have~$G\models\phi(v)\iff G\star_{\XXX}^{q+1}q\models\hphi(v)$. We
  choose~$\ell,m$ according to Theorem~\ref{thm:splittergame} such
  that Splitter has a winning strategy for the~$(\ell,m,2r)$-splitter
  game on every graph in~$\CCC$.  Note that~$q,r,\ell,m$ and~$\hphi$
  only depend on~$\phi$ and the class~$\CCC$, but not on~$\epsilon$ or
  the input graph~$G$. Now~$\epsilon$ comes into
  play. Let~$\delta=\epsilon/(2\ell)$. Choose $n_0=n_0(\delta,r)$
  according to Theorem~\ref{thm:alg-covers} such that every
  graph~$G\in\CCC$ of order~$n\ge n_0$ has an $r$-neighbourhood cover
  of radius at most~$2r$ and maximum degree at
  most~$n^{\delta}$. Choose~$n_1\ge n_0$ such that~$n_1^{\delta/2}\ge
  2$ and that every graph~$G\in\CCC$ of order~$n\ge n_1$ has at most
  $n^{1+\delta}$ edges. The existence of such an~$n_1$ follows from
  Lemma~\ref{lem:dens1}.  All the parameters and the
  formula~$\hphi(x)$ can be computed from~$\phi,\epsilon$ and the
  nowhere-density parameters of~$\CCC$ if~$\CCC$ is effectively
  nowhere dense.

  Now consider the~$\sigma$-coloured input graph~$G$. If~$n=|V(G)|< n_1$,
  we compute the set of all~$v\in V(G)$ such that~$G\models\phi(v)$ by
  brute force; in this case the running time can be bounded in terms
  of $\phi,\epsilon$, and~$\CCC$. So let us assume that~$n\ge n_1$. We
  compute an~$r$-neighbourhood cover~$\XXX$ of~$G$ of radius $2r$ and
  maximum degree~$n^{\delta}$. The main task of our algorithm will be
  to compute~$G\star_{\XXX}^{q+1}q$. Before we describe how to do
  this, let us assume that we have computed~$G\star_{\XXX}^{q+1}q$ and
  describe how the algorithm proceeds from there. The next step is to
  evaluate all~$(q,r)$-independence sentences in the Boolean
  combination~$\hphi(x)$ in~$G\star_{\XXX}^{q+1}q$. Consider such a
  sentence
  \[
  \psi=\exists x_1\ldots\exists x_q\Big(\bigwedge_{1\le i<j\le
    q}\dist(x_i,x_j)> 2r\wedge\bigwedge_{1\le i\le q} \chi(x_i)\Big).
  \]
  Remember that~$\chi(x_i)$ is an atomic formula. Thus we can easily
  compute the set~$U$ of all~$v\in V(G)$ such that
  $G\star_{\XXX}^{q+1}q\models\chi(v)$. Then we can use the algorithm
  of Theorem~\ref{thm:scattered-set} to decide if~$U$ has~$k$ elements
  of pairwise distance greater than~$2r$. This is the case if and only
  if $G\star_{\XXX}^{q+1}q\models\psi$.  This way, we decide
  which~$(q,r)$-independence sentences in~$\hphi(x)$ are satisfied
  in~$G\star_{\XXX}^{q+1}q$. It remains to evaluate the atomic
  formulas in~$\hphi(x)$ and combine the results to evaluate the
  Boolean combination. Both tasks are easy.

  Let us now turn to computing~$G\star_{\XXX}^{q+1}q$. We inductively
  compute~$G\star_{\XXX}^{i}q$ for~$0\le i\le q+1$. The base
  step~$i=0$ is trivial, because~$G\star_{\XXX}^{0}q=G$. As each
  $G\star_{\XXX}^{i}q$ is a~$\sigma'$ coloured graph for
  some~$\sigma'$ (to be precise,~$\sigma'=\sigma\star^iq$), it
  suffices to show how to compute~$G\star_{\XXX}q$ from~$G$. To do
  this, for each formula $\xi(x)\in\Phi^+(\sigma,1,q,q)$ we need to
  compute the set $P_\xi(G\star_{\XXX}q)$ of all~$v\in V(G)$ such that
  $G\big[\XXX(v)\big]\models\xi(v)$.  Let us fix a
  formula~$\xi(x)\in\Phi^+(\sigma,1,q,q)$.

  For every~$X\in\XXX$, let~$v_X\in X$ be a ``centre'' of~$G[X]$, that
  is, a vertex with~$X\subseteq N_{2r}(v_X)$. Such a~$v_X$ exists
  because the radius of~$G[X]$ is at most~$2r$. Let~$W_X\subseteq
  N_{2r}^G$ be Splitter's response if Connector chooses~$v_X$ in the
  first round of the~$(\ell,m,2r)$-splitter game on~$G$. Without loss
  of generality we assume that
  $W_X\neq\emptyset$. Let~$w_1,\ldots,w_m$ be an enumeration
  of~$W_X$. We apply Lemma~\ref{lem:redlem} with~$k=1$,~$\ell=q$,
  and~$m,q$ to the formulas~$\xi_0(x_1,y_1\ldots,y_m)=\xi(x_1)$ and
  $\xi_j(x_1,y_1\ldots,y_m)=\xi(y_j)$ for~$j=1,\ldots,m$.
  Let~$\sigma'$ be the vocabulary obtained by
  Lemma~\ref{lem:redlem}~(1), and let~$G_X$ be the graph obtained
  from~$G$ and~$w_1,\ldots,w_m$ by
  Lemma~\ref{lem:redlem}~(3). (Neither $\sigma'$ nor~$G_X$ depend on
  the formula.) For~$0\le j\le m$, let $\xi_j'(x_1)$ be the formula
  obtained from~$\xi_j$ by Lemma~\ref{lem:redlem}~(2).  We recursively
  evaluate the formulas~$\xi'_0,\ldots,\xi'_1$ in~$G_X$. This gives us
  the set $\Xi_X$ of all~$v\in V(G)$ such that
  $G\big[X\big]\models\xi(v)$. Doing this for all~$X\in\XXX$, we can
  compute the set
  \[
  P_\xi(G\star_{\XXX}q)=\{v\in V(G)\mid G[\XXX(v)]\models\xi(v)\}
  =\bigcup_{X\in\XXX}\big(\Xi_x\cap\{v\in V(G)\mid \XXX(v)=X\}\big).
  \]
  The crucial observation to ensure that the algorithm terminates is
  that in a recursive call with input~$G_X,\xi_j'$ the parameters~$q$
  and hence~$r=f_q(q)$ can be left unchanged. Moreover, it follows
  from the definition of~$G_X$ that Splitter has a winning strategy
  for the~$(\ell-1,m,2r)$-splitter game on~$G_X$. Thus we can reduce
  the parameter~$\ell$ by~$1$. Once we have reached~$\ell=0$, the
  graph~$G_X$ will be empty, and the algorithm terminates.

  There is one more issue we need to attend to, and that is how we
  compute Splitter's winning strategy, that is, the sets~$W_X$. We use
  Remark~\ref{rem:splitterrunningtime}. This means that to
  compute~$W_X$ in some recursive call, we need the whole history of
  the game (in a sense, the whole call stack). In addition, we need a
  breadth-first search tree in all graphs that appeared in the game
  before. It is no problem to compute a breadth-first search tree once
  when we first need it and then store it with the graph; this only
  increases the running time by a constant factor.

  This completes the description of the algorithm.

  Let us analyse the running time. The crucial parameters are the
  order $n$ of the input graph and the level~$j$ of the recursion. As
  argued above, we have~$j\le\ell$. We write the running time as a
  function~$T$ of~$j$ and~$n$. We first observe that the time used by
  the algorithm without the recursive calls can be bounded by~$c_1
  n^{1+\delta}$ for a suitable constant $c_1$ depending on the input
  sentence~$\phi$, the parameter $\epsilon$, and the class~$\CCC$, but
  not on~$n$ or~$j$. Furthermore, for~$n<n_1$ the running time can be
  bounded by a constant~$c_2$ that again only depends
  on~$\phi,\epsilon$, and~$\CCC$, and for~$j=0$ the running time can
  be bounded by~$c_3$. Furthermore, there is a $c_4$ such that for
  each~$X\in\XXX$ at most~$c_4$ recursive calls are made to the
  graph~$G_X$. Let~$n_X=|V(G_X)|\le |X|$ and $c=\max\{c_1,c_2,c_3,c_4\}$.
  We obtain the following recurrence for~$T$:
  \begin{align*}
    T(0,n)&\le c,\\
    T({j},n)&\le c&\text{for all }n< n_1,\\
    T({j},n)&\le\sum_{X\in\XXX}cT({j}-1,n_X)+cn^{1+\delta}&\text{for
      all }{j}\ge 1, n\ge n_1
  \end{align*}
  We claim that for all~$n\ge 1$ and~$0\le j\le\ell$ we have
  \begin{equation}\label{eq:mt1}
    T({j},n)\le c^jn^{1+2{j}\delta}=c^\ell n^{1+\epsilon}.
  \end{equation}
  As~$c$ and~$\ell$ are bounded in terms of~$\phi,\epsilon,\CCC$, this
  proves the theorem.

  \eqref{eq:mt1} can be proved by a straightforward induction. The
  crucial observation is
  \begin{equation}
    \label{eq:mt2}
    \sum_{X\in\XXX}n_X=\sum_{v\in V(G)}|\{X\in\XXX\mid v\in X\}|\le nn^\delta=n^{1+\delta}.
  \end{equation}
  The base steps
  ${j}=0$ and~$n<n_1$ are trivial. In the inductive step, we have
  \begin{align*}
    T({j},n)&\le\sum_{X\in\XXX} cT({j}-1,n_X)+cn^{1+\delta}\\
    &\le \sum_{X\in\XXX}
    cc^{j-1}n_X^{1+2({j}-1)\delta}+cn^{1+\delta}&\text{(Induction Hypothesis)}\\
    &\le c^j\Big(\sum_{X\in\XXX}n_X\Big)^{1+2({j}-1)\delta}+cn^{1+\delta}\\
    &\le c^jn^{(1+\delta)(1+2({j}-1)\delta)}+cn^{1+\delta}&\text{(by \eqref{eq:mt2})}\\
    &\le c^j\big(n^{1+(2{j}-1)\delta+2({j}-1)\delta^2}+n^{1+\delta}\big)\\
    &\le
    c^j\left(\frac{n^{1+2{j}\delta}+n^{1+(3/2)\delta}}{n^{\delta/2}}\right)&\text{(because
      ~$2({j}-1)\delta^2\le
      \frac{\epsilon^2}{2\ell}\le\delta/2$)}\\
    &\le c^jn^{2{j}\delta}&\text{(because~$n^{\delta/2}\ge 2$).}
  \end{align*}
\end{proof}

\section{Conclusion}
\label{sec:conclusion}
We prove that deciding first-order properties 
is fixed-parameter tractable on nowhere dense graph classes. This
generalises a long list of previous algorithmic meta theorems for
first-order logic. Furthermore, it is optimal on classes of graphs
closed under taking subgraphs. 
It remains open to find an optimal meta theorem for first-order
properties on classes that are not closed under taking subgraphs, but
only satisfy some weaker closure condition like being closed under
taking induced subgraphs. 

Our theorem underlines that nowhere dense graph classes have very favourable
algorithmic properties. As opposed to 
Robertson and Seymour's structure theory underlying most algorithms on
graph classes with excluded minors, the graph theory behind our
algorithms does not cause enormous
hidden constants in the running time. 

A particularly interesting property of nowhere dense classes and
classes of bounded expansion that we
uncover here for the first time is that they have simple sparse neighbourhood
covers with very good parameters. We have focussed on the radius of the
covering sets and have not tried to optimise the degree of the cover, that
is, the number of covering sets a vertex may be contained in. As
the graph theory underlying our result is not very complicated, we
believe that it is possible to obtain good degree bounds as well, probably
much better than those obtained through graph minor theory~\cite{abrgavmal+07,buscoslaf+07} (even
though the classes we consider are much larger). However, this remains
future work.

%
\bibliographystyle{plain}
\bibliography{nowheredense}  

\begin{thebibliography}{10}

\bibitem{abrgavmal+07}
I.~Abraham, C.~Gavoille, D.~Malkhi, and U.~Wieder.
\newblock Strong-diameter decompositions of minor free graphs.
\newblock In {\em Proceedings of the nineteenth annual ACM Symposium on
  Parallel Algorithms and Architectures}, pages 16--24, 2007.

\bibitem{bodfomlok+09}
H.L. Bodlaender, F.V. Fomin, D.~Lokshtanov, E.~Penninkx, S.~Saurabh, and D.M.
  Thilikos.
\newblock ({M}eta) {K}ernelization.
\newblock In {\em Proceedings of the 50th Annual IEEE Symposium on Foundations
  of Computer Science}, pages 629--638, 2009.

\bibitem{buscoslaf+07}
C.~Busch, R.~LaFortune, and S.~Tirthapura.
\newblock Improved sparse covers for graphs excluding a fixed minor.
\newblock In {\em Proceedings of the twenty-sixth annual ACM Symposium on
  Principles of Distributed Computing}, pages 61--70, 2007.

\bibitem{cou90}
B.~Courcelle.
\newblock Graph rewriting: An algebraic and logic approach.
\newblock In J.~van Leeuwen, editor, {\em Handbook of Theoretical Computer
  Science}, volume~B, pages 194--242. Elsevier Science Publishers, 1990.

\bibitem{coumakrot00}
B.~Courcelle, J.A. Makowsky, and U.~Rotics.
\newblock Linear time solvable optimization problems on graphs of bounded
  clique width.
\newblock {\em Theory of Computing Systems}, 33(2):125--150, 2000.

\bibitem{coumakrot01}
B.~Courcelle, J.A. Makowsky, and U.~Rotics.
\newblock On the fixed-parameter complexity of graph enumeration problems
  definable in monadic second-order logic.
\newblock {\em Discrete Applied Mathematics}, 108(1--2):23--52, 2001.

\bibitem{dawgrokre07}
A.~Dawar, M.~Grohe, and S.~Kreutzer.
\newblock Locally excluding a minor.
\newblock In {\em Proceedings of the 22nd IEEE Symposium on Logic in Computer
  Science}, pages 270--279, 2007.

\bibitem{dawgrokre+06}
A.~Dawar, M.~Grohe, S.~Kreutzer, and N.~Schweikardt.
\newblock Approximation schemes for first-order definable optimisation
  problems.
\newblock In {\em Proceedings of the 21st IEEE Symposium on Logic in Computer
  Science}, pages 411--420, 2006.

\bibitem{DawarKre09}
A.~Dawar and S.~Kreutzer.
\newblock Domination problems in nowhere-dense classes of graphs.
\newblock In {\em Foundations of software technology and theoretical computer
  science (FSTTCS)}, pages 157--168, 2009.

\bibitem{Diestel05}
R.~Diestel.
\newblock {\em Graph Theory}.
\newblock Springer-Verlag, 3rd edition, 2005.

\bibitem{DowneyF98}
R.~Downey and M.~Fellows.
\newblock {\em Parameterized Complexity}.
\newblock Springer, 1999.

\bibitem{dowfeltay96}
R.G. Downey, M.R. Fellows, and U.~Taylor.
\newblock The parameterized complexity of relational database queries and an
  improved characterization of {W}[1].
\newblock In D.S. Bridges, C.~Calude, P.~Gibbons, S.~Reeves, and I.H. Witten,
  editors, {\em Combinatorics, Complexity, and Logic}, volume~39 of {\em
  Proceedings of DMTCS}, pages 194--213. Springer-Verlag, 1996.

\bibitem{Dvorak13}
Z.~Dvo{\v r}{\'a}k.
\newblock Constant-factor approximation of the domination number in sparse
  graphs.
\newblock {\em Eur. J. Comb.}, 34(5):833--840, 2013.

\bibitem{dvokratho10}
Z.~Dvo{\v r}{\'a}k, D.~Kr{\'a}l, and R.~Thomas.
\newblock Deciding first-order properties for sparse graphs.
\newblock In {\em Proceedings of the 51st Annual IEEE Symposium on Foundations
  of Computer Science}, pages 133--142, 2010.

\bibitem{EbbinghausFluTho94}
H.-D. Ebbinghaus, J.~Flum, and W.~Thomas.
\newblock {\em Mathematical Logic}.
\newblock Springer, 2nd edition, 1994.

\bibitem{flugro01}
J.~Flum and M.~Grohe.
\newblock Fixed-parameter tractability, definability, and model checking.
\newblock {\em {SIAM} Journal on Computing}, 31(1):113--145, 2001.

\bibitem{FlumGro06}
J.~Flum and M.~Grohe.
\newblock {\em Parameterized Complexity Theory}.
\newblock Springer Verlag, 2006.

\bibitem{frigro01}
M.~Frick and M.~Grohe.
\newblock Deciding first-order properties of locally tree-decomposable
  structures.
\newblock {\em Journal of the ACM}, 48:1184--1206, 2001.

\bibitem{Gaifman82}
H.~Gaifman.
\newblock On local and non-local properties.
\newblock In J.~Stern, editor, {\em Herbrand Symposium, Logic Colloquium '81},
  pages 105 -- 135. North Holland, 1982.

\bibitem{gro07b}
M.~Grohe.
\newblock Logic, graphs, and algorithms.
\newblock In J.~Flum, E.~Gr\"adel, and T.~Wilke, editors, {\em Logic and
  Automata -- History and Perspectives}, volume~2 of {\em Texts in Logic and
  Games}, pages 357--422. Amsterdam University Press, 2007.

\bibitem{grokawree13}
M.~Grohe, K.~Kawarabayashi, and B.~Reed.
\newblock A simple algorithm for the graph minor decomposition -- logic meets
  structural graph theory--.
\newblock In {\em Proceedings of the 24th Annual ACM-SIAM Symposium on Discrete
  Algorithms}, 2013.

\bibitem{grokre11}
M.~Grohe and S.~Kreutzer.
\newblock Methods for algorithmic meta theorems.
\newblock In M.~Grohe and J.A. Makowsky, editors, {\em Model Theoretic Methods
  in Finite Combinatorics}, volume 558 of {\em Contemporary Mathematics}, pages
  181--206. American Mathematical Society, 2011.

\bibitem{KiersteadY03}
H.A. Kierstead and D.~Yang.
\newblock Orderings on graphs and game coloring number.
\newblock {\em Order}, 20(3):255--264, 2003.

\bibitem{kre11}
S.~Kreutzer.
\newblock Algorithmic meta-theorems.
\newblock In J.~Esparza, C.~Michaux, and C.~Steinhorn, editors, {\em Finite and
  Algorithmic Model Theory}, London Mathematical Society Lecture Note Series,
  chapter~5, pages 177--270. Cambridge University Press, 2011.

\bibitem{kretaz10a}
S.~Kreutzer and S.~Tazari.
\newblock Lower bounds for the complexity of monadic second-order logic.
\newblock In {\em Proceedings of the 25th IEEE Symposium on Logic in Computer
  Science}, 2010.

\bibitem{kretaz10}
S.~Kreutzer and S.~Tazari.
\newblock On brambles, grid-like minors, and parameterized intractability of
  monadic second-order logic.
\newblock In {\em Proceedings of the 21st Annual ACM-SIAM Symposium on Discrete
  Algorithms}, 2010.

\bibitem{LangerRRS12}
A.~Langer, F.~Reidl, P.~Rossmanith, and S.~Sikdar.
\newblock Evaluation of an mso-solver.
\newblock In {\em ALENEX}, pages 55--63, 2012.

\bibitem{NesetrilOdM12}
J.~{Ne\v set\v ril} and P.~Ossona {de Mendez}.
\newblock {\em Sparsity}.
\newblock Springer, 2012.

\bibitem{NesetrilO11}
J.~{Ne\v{s}et{\v r}il} and P.~{Ossona de Mendez}.
\newblock On nowhere dense graphs.
\newblock {\em European Journal of Combinatorics}, 32(4):600--617, 2011.

\bibitem{NesetrilO05}
J.~Nešetřil and P.~Ossona de~Mendez.
\newblock Grad and classes with bounded expansion ii. algorithmic aspects.
\newblock {\em arXiv}, math/0508324v2, 2005.

\bibitem{papyan91}
C.H. Papadimitriou and M.~Yannakakis.
\newblock Optimization, approximation, and complexity classes.
\newblock {\em Journal of Computer and System Sciences}, 43:425--440, 1991.

\bibitem{see96}
D.~Seese.
\newblock Linear time computable problems and first-order descriptions.
\newblock {\em Mathematical Structures in Computer Science}, 6:505--526, 1996.

\bibitem{Zhu09}
X.~Zhu.
\newblock Colouring graphs with bounded generalized colouring number.
\newblock {\em Discrete Mathematics}, 309:5562--5568, 2009.

\end{thebibliography}
%
%

\end{document}